\documentclass[sigconf,nonacm]{acmart}
\AtBeginDocument{%
  }

\setcopyright{none}
\copyrightyear{2026}
\acmYear{2026}
\acmDOI{}
\acmConference[CCS '26]{Proceedings of the 2026 ACM SIGSAC Conference on Computer and Communications Security}{November 15--19, 2026}{The Hague, The Netherlands}
\acmISBN{}

\renewcommand\footnotetextcopyrightpermission[1]{}

\usepackage{xcolor}
\usepackage{courier} 
\usepackage{balance} 
\usepackage{url}

\usepackage{wrapfig}
\usepackage{graphicx}
\usepackage{textcomp}

\usepackage{subcaption}
\usepackage{comment}

\usepackage{tikz}
\usepackage{amsfonts}
\usepackage{mathrsfs}
\usepackage{makecell}
\usepackage{bm}
\usepackage{amsmath,amsfonts,amsthm,mathtools} 
\usepackage{bbm}
\usepackage[framemethod=TikZ]{mdframed}
\usepackage{hyperref}

\usepackage{booktabs}
\usepackage{multicol}
\usepackage{multirow}
\usepackage{colortbl}
\usepackage{tcolorbox}
\usepackage{float}

\usepackage{soul}
\usepackage{tabularx}
\usepackage{xspace}
\usepackage{bbding}
\usepackage{pifont}

\usepackage{wasysym}
\usepackage{threeparttable}

\usepackage[linesnumbered,ruled,vlined]{algorithm2e} 
\usepackage{cleveref}
\usepackage{acronym}

\newenvironment{packeditemize}{
\begin{list}{$\bullet$}{
\setlength{\labelwidth}{8pt}
\setlength{\itemsep}{2pt}
\setlength{\leftmargin}{\labelwidth}
\addtolength{\leftmargin}{\labelsep}
\setlength{\parindent}{0pt}
\setlength{\listparindent}{\parindent}
\setlength{\parsep}{0pt}
\setlength{\topsep}{2pt}}}{\end{list}}

\newenvironment{packedenumerate}{
\begin{list}{\arabic{enumi}.}{
\usecounter{enumi}
\setlength{\labelwidth}{10pt}
\setlength{\itemsep}{2pt}
\setlength{\leftmargin}{\labelwidth}
\addtolength{\leftmargin}{\labelsep}
\setlength{\parindent}{0pt}
\setlength{\listparindent}{\parindent}
\setlength{\parsep}{0pt}
\setlength{\topsep}{2pt}}}
{\end{list}}

\mdfdefinestyle{customstyle}{%
    backgroundcolor=white,%
    linewidth=1pt,
    frametitlefont=\bfseries,%
    roundcorner=5pt,
    innertopmargin=5pt,
    innerbottommargin=5pt,
    nobreak=true
}

\acrodef{tee}[TEE]{Trusted Execution Environment}

\newcommand{\boxname}{{\textsc{DIPBox}}\xspace}

\newcommand{\vicf}{f_\cV}
\newcommand{\vicdst}{\cX_\cV}
\newcommand{\vicdistri}{\cP_\cV}

\newcommand{\auditset}{\cX_\cJ}

\newcommand{\advf}{f_\cA}
\newcommand{\advdst}{\cX_\cA}
\newcommand{\advdistriaccess}{\cP_\cA^{(0)}}

\newcommand{\advdistri}{\cP_\cA}

\renewcommand{\epsilon}{\varepsilon}

\def\:#1{\protect \ifmmode {\mathbf{#1}} \else {\textbf{#1}} \fi}

\newcommand{\bx}{\mathbf{x}}

\newcommand{\cA}{\mathcal{A}}
\newcommand{\cB}{\mathcal{B}}

\newcommand{\cJ}{\mathcal{J}}

\newcommand{\cM}{\mathcal{M}}

\newcommand{\cP}{\mathcal{P}}

\newcommand{\cV}{\mathcal{V}}

\newcommand{\cX}{\mathcal{X}}

\renewcommand{\epsilon}{\varepsilon}

\DeclareMathOperator*{\expectedvalue}{\mathbb{E}}

\newcommand{\norm}[1]{\left\lVert#1\right\rVert}
\newcommand{\abs}[1]{\left\lvert#1\right\rvert}

\theoremstyle{definition}
\newtheorem{definition}{Definition}[section]
\newtheorem{theorem}{Theorem}[section]
\newtheorem{lemma}[theorem]{Lemma}
\newtheorem{proposition}{Proposition}

\newcommand{\eg}{\hbox{{e.g.}}\xspace}

\newcommand{\ie}{\hbox{{i.e.}}\xspace}

\newcommand{\aka}{\hbox{{a.k.a.}}\xspace}

\newcommand{\bheading}[1]{{\vspace{4pt}\noindent{\textbf{#1}}}}
\newcommand{\iheading}[1]{{\vspace{4pt}\noindent{\textit{#1}}}} 
\newcommand{\redbox}[1]{\colorbox{red!20}{#1}}
\newcommand{\greenbox}[1]{\colorbox{green!20}{#1}}

\begin{document}

\title{\boxname: A Multi-scale Testing Framework for Tracking \\Dataset Regeneration}
\author{Tian Dong$^{\ast}$, Yan Meng$^{\ast}$, Shaofeng Li$^{\dagger}$, Guoxing Chen$^{\ast}$, Yuling Chen$^{\ddagger}$, Zhen Liu$^{\ast}$, Haojin Zhu$^{\ast}$, \and 
Hao Chen$^{\S}$}
\affiliation{
\institution{$^{\ast}$Shanghai Jiao Tong University, China}
\institution{$^{\dagger}$Southeast University, China}
\institution{$^{\ddagger}$Guizhou University, China}
\institution{$^{\S}$The University of Hong Kong, China}
\country{}
}

\begin{abstract}
Training datasets have tremendous proprietary value and are vulnerable to unauthorized copying.
Existing defenses mainly focus on tracking individual data points, but pay little attention to the threat of dataset regeneration.
Through a measurement study of public tumor datasets, we identify substantial real-world partial-dataset replication, raising concerns about potential license noncompliance.
To counter the challenge of tracking previously unknown adversarial regeneration, our key insight is that regeneration that preserves model utility inevitably preserves measurable signals across multiple feature scales.
We categorize these dataset features into sample-, set-, and distribution-level features and design four similarity metrics to accurately identify regeneration.
Based on these metrics, we develop \boxname, which to our knowledge is the first testing framework that tracks regeneration suspects via multi-scale similarity testing across a spectrum of defender access settings, from limited to full information.
We further provide a learning-theoretic analysis that justifies these multi-scale metrics and formalizes an inherent utility--divergence trade-off, implying fundamental limits on evasive regeneration.
Extensive experiments on 16 vision and text base datasets, 320 regenerated datasets, and 590 derived models validate that \boxname outperforms previous solutions while characterizing its robustness and limits under three adaptive attacks.

\end{abstract}




\maketitle

\section{Introduction}

Training datasets are \textit{copyrightable} assets because of the high cost of curation and the increasing scarcity of high-quality data sources~\cite{exhaust_data_24}, but they are routinely replicated and used without the owner’s consent, leading to potential copyright or license violations~\cite{llama_leak_23}.
For example, more than 8{,}284 open-source datasets on Hugging Face use Creative Commons Non-Commercial licenses, among which 1{,}621 also prohibit derivatives, but practical compliance mechanisms are lacking~\cite{DBLP:journals/corr/abs-2310-16787}.
Even worse, the absence of an explicit license for over 70\% of datasets on GitHub and Hugging Face complicates legal accountability~\cite{DBLP:journals/corr/abs-2310-16787}.

\begin{figure}[t] 
    \centering
    \includegraphics[width=\linewidth]{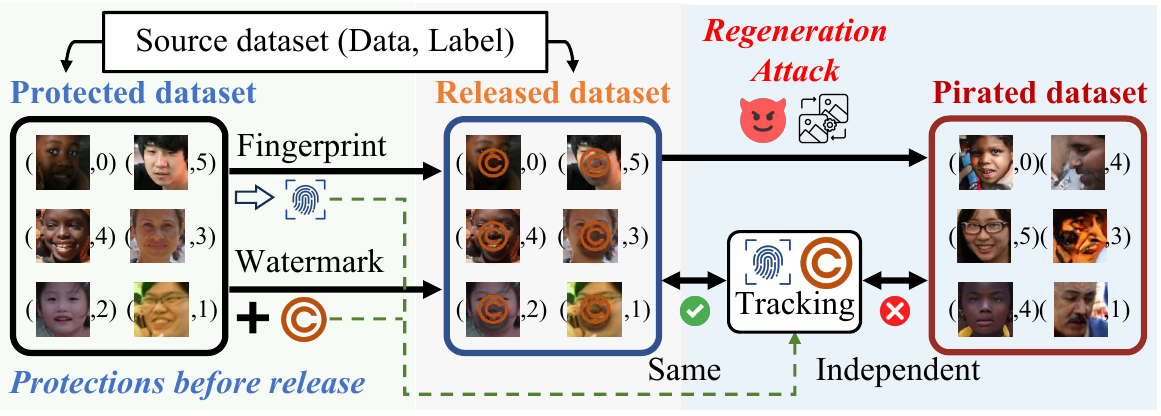}
    \caption{The adversary regenerates a derived dataset to evade detection and gain economic benefit.}
    \label{fig:regeneration}
\end{figure}

Existing dataset tracking approaches can be categorized into \emph{proactive} and \emph{passive} defenses.
Proactive defenses watermark datasets prior to release via backdoors~\cite{li2022untargeted,yimingli_tifs23} or feature perturbations~\cite{DBLP:conf/icml/SablayrollesDSJ20,datause_ccs2024}.
Passive defenses, which apply to in-the-wild datasets, instead exploit model memorization to infer whether a model was trained on a target dataset in whole or in part~\cite{maini2021dataset,huang2022dataset,DBLP:journals/tifs/LiuXMW22,DBLP:conf/ndss/DongLCXZ023,dataset_auditing_rl}.
A key limitation of prior art is its reliance on memorization of individual data points to track a dataset in its entirety, together with the assumption that adversaries directly reuse or at most mildly modify the pirated data.
This leaves room for real-world adversaries to evade detection via \textit{regeneration attacks} (see \Cref{fig:regeneration}), including aggressive per-datum modification, set regrouping~\cite{DBLP:journals/corr/abs-2310-16787}, or generative reproduction (\eg, using Stable Diffusion) that yields no identical samples while preserving equivalent model-training utility~\cite{wang2024do}.

\begin{table*}[t]
\centering
\caption{Comparison with prior work on dataset regeneration tracking.
In the suspect-dataset access column, \CIRCLE denotes no dataset access and \LEFTcircle denotes access to a subset.
In the suspect-model access column, \CIRCLE denotes black-box query access and \Circle denotes white-box access.
Commas indicate support for multiple settings.
}
\label{tab:taxonomy}
\resizebox{\textwidth}{!}{
\begin{tabular}{cccccccc}
\toprule
\multirow{2}{*}{\textbf{Methods}} & \multirow{2}{*}{\textbf{Type}} & \multirow{2}{*}{\textbf{Non-Invasiveness}} & \multirow{2}{*}{\begin{tabular}[c]{@{}c@{}}\textbf{Access to} \\ \textbf{Suspect Dataset $\advdst$}\end{tabular}} & \multirow{2}{*}{\begin{tabular}[c]{@{}c@{}}\textbf{Access to} \\ \textbf{Suspect Model} $\advf$ \end{tabular}} & \multicolumn{3}{c}{\textbf{Regeneration Attacks}}                                                            \\ \cline{6-8} 
 &           &                                   &                                 &                               & \multicolumn{1}{c}{\ding{182} Post-processing} & \multicolumn{1}{c}{\ding{183} Reorganization} & \ding{184} Polishing \\ \midrule

\rowcolor{gray!15}
Radioactive data~\cite{DBLP:conf/icml/SablayrollesDSJ20}, DVBW~\cite{yimingli_tifs23} &        Watermarking  &   \ding{55}         &   \CIRCLE         & \Circle,\CIRCLE                    & \multicolumn{1}{c}{\ding{55}}                    & \multicolumn{1}{c}{\ding{55}}                   &   \ding{55}   \\ \hline
Data-use Audit~\cite{datause_ccs2024}   &   Watermarking  &    \ding{55}       &      \CIRCLE          &      \CIRCLE             & \multicolumn{1}{c}{\ding{51}}                    & \multicolumn{1}{c}{\ding{51}}                   &    \ding{55} \\ \hline   \rowcolor{gray!15}

Dataset Inference (DI)~\cite{maini2021dataset,maini2024llm} &        Inference  &   \ding{51}         &   \CIRCLE         & \Circle,\CIRCLE                    & \multicolumn{1}{c}{\ding{55}}                    & \multicolumn{1}{c}{\ding{55}}                   &   \ding{55}   \\ \hline 
EMA~\cite{huang2022dataset}, MeFA~\cite{DBLP:journals/tifs/LiuXMW22}, RAI$^2$~\cite{DBLP:conf/ndss/DongLCXZ023}      &  Inference     &   \ding{51}         &         \CIRCLE          &       \CIRCLE      & \multicolumn{1}{c}{\ding{51}}                    & \multicolumn{1}{c}{\ding{51}}                   & \ding{55}
\\ \hline \rowcolor{gray!15}
ORL-AUDITOR~\cite{dataset_auditing_rl}      &  Inference     &   \ding{51}         &         \CIRCLE           &       \CIRCLE      & \multicolumn{1}{c}{\ding{55}}                    & \multicolumn{1}{c}{\ding{51}}                   & \ding{55}    
\\ \hline

\textbf{\boxname~(Ours)}   &    Similarity Testing    &               \ding{51}  &       \LEFTcircle,\CIRCLE     &                \Circle,\CIRCLE  & \multicolumn{1}{c}{\ding{51}}                    & \multicolumn{1}{c}{\ding{51}}                   &   \ding{51}  \\ \bottomrule
\end{tabular}
}

\end{table*}

To validate this threat, we first conduct a measurement study of publicly available tumor datasets and identify substantial duplication that suggests potential license noncompliance.
For example, a dataset licensed to restrict derivatives and commercial use appears to have been repackaged into four other open-source datasets, with an exact-duplicate overlap rate ranging from 49\% to 100\%.
These measurements confirm the existence of regenerated datasets in practice and motivate our research question: \textit{How can we robustly track dataset regeneration for datasets released on public platforms?}

Similar to other digital artifacts, the core challenge lies in handling varied and unforeseen regeneration strategies under the constraints of limited access (\eg, a pirated trial subset) or no access (\eg, only a trained model) to the full pirated dataset.
To tackle this challenge, we observe that, regardless of efforts to obfuscate regeneration, an adversary typically seeks to preserve dataset quality so that equally accurate models can be trained—an objective that resembles software infringement.
Consequently, inspired by software copyright protections that compare multiple components~\cite{register_software, software_copyright_explained} (\eg, code, functionality, interfaces), we adopt a multi-scale similarity testing approach that delivers robustness and extensibility to new tracking metrics.

This paper introduces \boxname, which to our knowledge is the first framework designed for third parties (\eg, responsible sharing platforms) to track dataset regeneration via multi-scale testing against the source.
We identify three core feature levels targeted by regeneration attacks: sample-level, set-level (shallow features), and distribution-level (deep features).
To \emph{quantify} feature discrepancy, in addition to existing metrics such as \textit{output distance}~\cite{DBLP:conf/ndss/DongLCXZ023,dataset_auditing_rl} and \textit{sample distance} that primarily capture shallow features, we introduce two new metrics—\textit{gradient distance} and \textit{distribution distance}—to generalize testing across scales.
We further design metric-adaptive test-case selection and a principled decision procedure, covering a broader spectrum of regeneration attacks under wider access settings compared with prior work (see \Cref{tab:taxonomy}) and provide a learning-theoretic analysis (see \Cref{sec:theory}) that (i) justifies why utility-preserving regeneration retains detectable signals, and (ii) formalizes a utility-divergence trade-off for evasive regeneration.

We implement \boxname as an automated toolbox to support future measurement and defense research.
To evaluate, we build a regeneration benchmark covering 320 regenerated datasets derived from 16 widely used base datasets and 590 trained models spanning 12 classic and foundational model architectures.
The results show that \boxname provides accurate, evidence-based identification of regenerated datasets, remains model-agnostic, and is resilient to diverse adversarial training conditions and regeneration attacks.
For instance, \boxname can track all distribution-level regeneration variants of FairFace containing $128\times 128$ facial images.
Additionally, \boxname surpasses state-of-the-art methods~\cite{datause_ccs2024,DBLP:conf/ndss/DongLCXZ023} with at least a 25\% higher TPR, due to its holistic, multi-scale characterization of dataset features.
Finally, we show that in our evaluated adaptive attacks, evading \boxname requires substantial utility loss or additional computation and memory cost.

\bheading{Contribution.}
In summary, our contributions are:
\begin{packeditemize}
\item We conduct a measurement study of real-world datasets and identify substantial duplication that suggests potential license noncompliance.
\item We formalize the problem of tracking dataset regeneration and present attack primitives.
\item We propose \boxname, which to our knowledge is the first extensible framework for identifying dataset regeneration based on multi-scale testing metrics with a learning-theoretic analysis.
\item We conduct extensive evaluations and implement \boxname as an automated toolbox to support future research.
\end{packeditemize}

\section{Motivating Measurement}
\label{subsec:measure-method}
We conduct a measurement of real-world datasets and identify potentially regenerated datasets.
We manually screened over 74 datasets with keywords ``brain'' and ``tumor'' on Hugging Face to ensure they target the same domain and task.
We filter for datasets with license ``CC-BY-NC-ND-4.0'', which restricts commercial use and derivative works.
Then, we select datasets with less restrictive or unspecified licenses using the following criteria:
(i) the dataset has been downloaded at least 10 times;
(ii) the dataset is not described as derived from existing datasets; and
(iii) the dataset is downloadable and usable.
In the end, we gathered 8 datasets (see \Cref{tab:measure-dataset}), annotated from ``D0'' to ``D7'', where ``D0'' and ``D1'' have licenses forbidding commercial use and derivatives, while the rest are more permissive or unspecified.
We unify the data size to $256\times 256$ and convert all images to grayscale.

\begin{table}[t]
\caption{Statistics of tumor datasets in the measurement.}
\label{tab:measure-dataset}
\resizebox{\linewidth}{!}{
\begin{tabular}{@{}ccccc@{}}
\toprule
\textbf{Dataset} & \textbf{Size} & \textbf{Classes} & \textbf{License} & \textbf{Repository ID} \\ \midrule
D0 & 275 & 10 & CC-BY-NC-ND-4.0 & TrainingDataPro/brain-mri-dataset \\
D1 & 447 & 4 & CC-BY-NC-ND-4.0 & haydenbanz/TumorVisionDatasets \\
D2 & 800 & 4 & Apache-2.0 & thubZ9/MRI\_Classification-Tumor \\
D3 & 506 & / & / & miladfa7/Brain-MRI-Images-for-Brain-Tumor-Detection \\
D4 & 7342 & 3 & / & youngp5/BrainMRI \\
D5 & 1429 & 2 & / & tanzuhuggingface/brainmri \\
D6 & 3264 & 4 & MIT & sartajbhuvaji/Brain-Tumor-Classification \\
D7 & 7023 & 4 & / & Simezu/brain-tumour-MRI-scan \\ \bottomrule
\end{tabular}
}

\end{table}

\begin{figure}[t]
    \centering
    \begin{subfigure}{0.24\columnwidth} 
        \includegraphics[width=\linewidth]{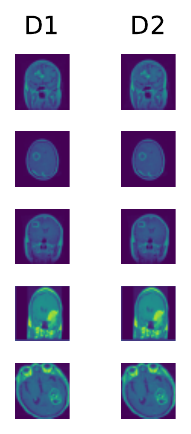}
        \caption{Exact duplicate examples.}
        \label{subfig:measureexample}
    \end{subfigure}
    \hfill
    \begin{subfigure}{0.7\columnwidth}
        \includegraphics[width=\linewidth]{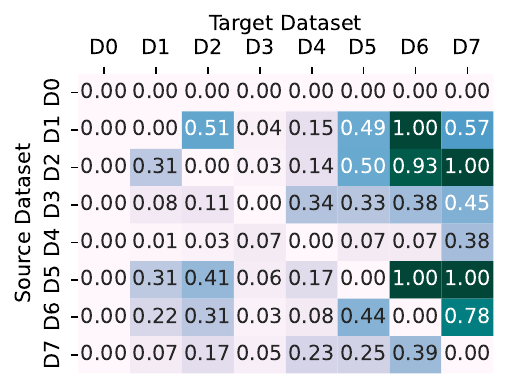}
        \caption{Proportion of common samples in the source dataset among the measured datasets.}
        \label{subfig:measureheatmap}
    \end{subfigure}
    \caption{Measurement on open-source brain tumor datasets.}
\end{figure}

We observe exact duplicate samples among the collected datasets.
\Cref{subfig:measureexample} presents examples of common samples ($L_2$ distance is 0) between two datasets D1 and D2.
Furthermore, for all datasets, we identify candidate duplicates by thresholding $L_2$ distance at 0.1 and manually verify the resulting pairs (see \Cref{appendix:measurementdetails} for verification details).
In \Cref{subfig:measureheatmap}, each grid shows the proportion of duplicate samples in the source dataset.
There are three notable findings:
\begin{enumerate}
    \item D0 has no overlapping samples with the measured datasets, which is also verified through manual checking.
Considering that D0 is a trial subset for attracting users to purchase the complete dataset, it may originate from a different data source than the commonly used tumor data.
\item Despite the more restrictive license, D1 dataset might be composed of the existing samples (\eg, subset of D6), because D1 was uploaded later than the other datasets.
This might raise compliance and attribution concerns.
\item Dataset regeneration is more prevalent among datasets with more permissive or unspecified licenses.
This not only confuses users about data sources but also causes potential disputes in the future.
\end{enumerate}

These findings suggest that real-world dataset regeneration can arise as sample-level or set-level reuse (Attack~\ding{182} and \ding{183} defined in \Cref{sec:threat_model}) and motivate us to investigate whether existing defenses remain robust under such regeneration primitives.
For example, since D0 is a trial subset used to attract users, an adversary may purchase, regenerate and resell the full dataset.

\section{Problem Statement}
\label{sec:threat_model}

In this section, we formalize the problem of tracking regenerated datasets and present the adversary's goals and capacities.

\subsection{Formulation \& Threat Model}
\label{subsec:formulation}
We focus on small-scale \textit{domain datasets}, typically used for training specialized classification models~\cite{zhao2023fedprompt} and widely used across the AI sector~\cite{practical_aisec_24}.
Constructing high-quality domain datasets is resource-intensive, with data source quality being crucial for ensuring model performance with real-world test data~\cite{DBLP:journals/corr/abs-2310-16787,martens2018importance}.
Consider an example of Electroencephalography (EEG) datasets: 
a medical company using up-to-date infrastructure from its clients' EEG data developed the FACED~\cite{chen2023large} dataset and achieved a 77.6\% accuracy for emotion recognition. 
In contrast, competitors without access to FACED but using similar yet lower-quality EEG datasets like SEED-IV~\cite{zheng2018emotionmeter} and DEAP~\cite{koelstra2011deap} could only achieve up to 27.1\% accuracy (Appendix~\ref{appendix:details_and_results}), highlighting the need for dataset tracking in the event of data breaches.

We consider three parties: the victim $\cV$, the adversary $\cA$, and the defender (\aka, auditor).
$\cV$ owns a proprietary dataset $\vicdst$.
The victim also trains a model $\vicf$ on $\vicdst$.
The value of $\vicdst$ lies in the fact that its distribution $\vicdistri$ closely matches the real-world distribution, leading to models that achieve higher accuracy on live test data $\cP_t$ after deployment, \ie, $\vicdst\sim\vicdistri\approx\cP_t$.
For the adversary $\cA$, directly collecting in-distribution data is infeasible due to exclusive data sources~\cite{bansal2022systematic} or high cost.
We assume $\cA$ has access only to a distribution $\advdistriaccess$ not identically distributed as $\cP_t$~\cite{DBLP:journals/corr/abs-1806-00582}, thus the model $f_{\advdistriaccess}$ trained on $\advdistriaccess$ achieves at least $\epsilon^u_{\cV,\cA}$ lower test accuracy:
\begin{equation}
    \expectedvalue_{(\bx, y)\sim\cP_t}[v(f_{\advdistriaccess}(\bx), y)] < \expectedvalue_{(\bx, y)\sim\cP_t}[v(\vicf(\bx), y)] - \epsilon^u_{\cV,\cA},
\end{equation}
where $v$ is the utility function and $\mathbb{E}$ is the expectation.

We denote the adversary's auxiliary dataset drawn from $\advdistriaccess$ as $\advdst^{(0)}\sim\advdistriaccess$.

We assume the adversary obtains $\vicdst$ through leakage or other unauthorized means, and then crafts and distributes $\advdst$ on platforms for profit.
The adversary also knows the existing protections~\cite{li2022untargeted, yimingli_tifs23,datause_ccs2024}, thus cannot directly reuse $\vicdst$; instead, they regenerate their dataset $\advdst$ (of distribution $\advdistri$) based on $\vicdst$ and train their model $f_\cA$.
The adversary can perform hyperparameter tuning to maximize the accuracy of $f_\cA$.
Formally, we model the regeneration and tracking as an adversarial game:

\begin{packedenumerate}
    \item Initially, the adversary has $\advdst^{(0)}$ and $\vicdst$.
    
    \item The adversary regenerates $\advdst\leftarrow\cA(\vicdst, \advdst^{(0)})$.
    
    \item The defender calculates $d_{\cV,\cA}\leftarrow \text{DIST}(\vicdst, \advdst)$.
    
    \item The adversary wins if $d_{\cV,\cA} > \tau_{\cV,\cA}$ and achieves:
    \begin{equation}
    \label{eq:game}
    \abs{\expectedvalue_{(\bx, y)\sim\cP_t}[v(f_\cA(\bx), y)] - \expectedvalue_{(\bx, y)\sim\cP_t}[v(f_\cV(\bx), y)]}<\epsilon^u_{\cV,\cA},
    \end{equation}
    where $v$ is the dataset utility function, $\epsilon^u_{\cV,\cA}$ is the utility drop limit, $\tau_{\cV,\cA}$ is the predefined threshold by $\cV$. 
\end{packedenumerate}

In this formulation, $\text{DIST}(\cdot,\cdot)$ is a generic distance-based auditing procedure; \boxname instantiates it using multiple distances and a joint decision rule (see \Cref{sec:system}).
The adversary aims to win the game by:
\begin{packeditemize}
    \item \textit{Preserving Utility}: $f_\cA$ has comparable accuracy on $\cP_t$.
    
    \item \textit{Evading Defense}: $\advdst$ cannot be judged as a copy of $\vicdst$.
\end{packeditemize}

\begin{figure}[t] 
    \centering
    \includegraphics[width=0.95\linewidth]{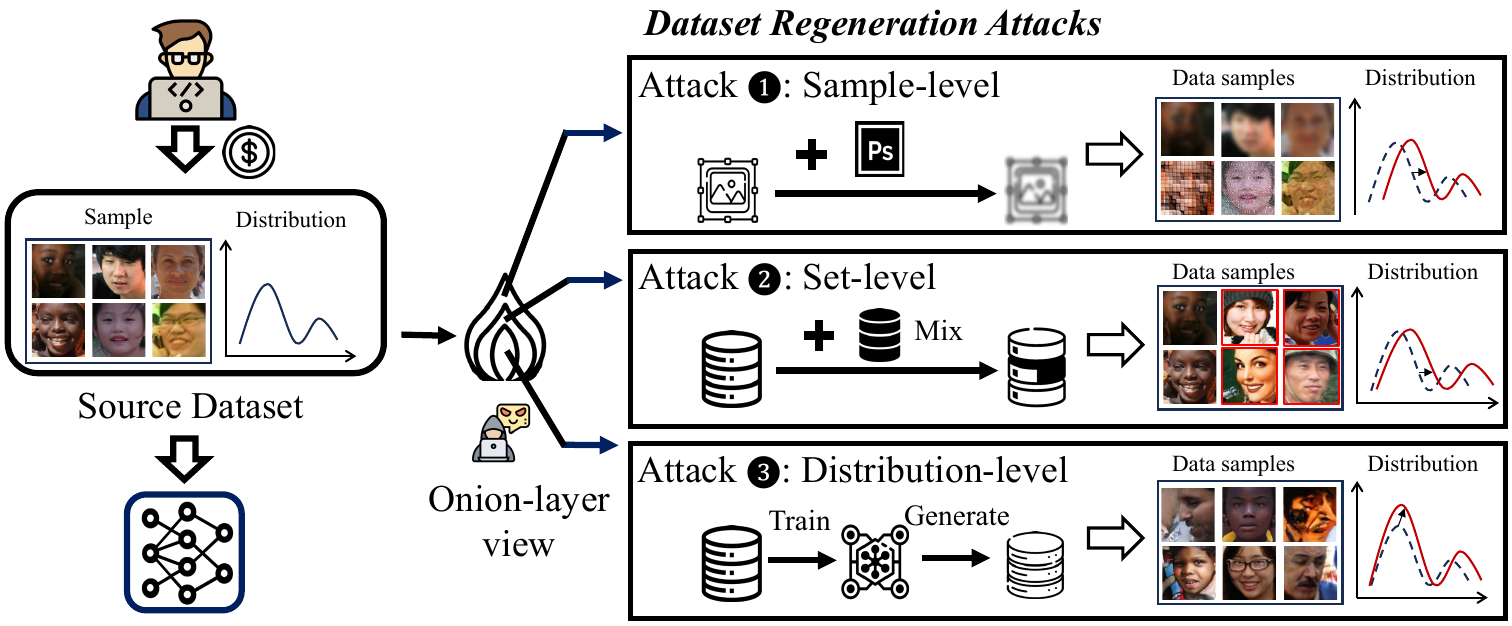}
    \caption{Overview of the regeneration attack primitives on our identified dataset features.
    The dataset has an onion-layered structure, where the surface-level feature (\eg, datum) is the skin and the inner feature (\ie, distribution) represents the onion core.
    The adversary can construct a pirated dataset by regeneration attacks at three levels.
    }
    \label{fig:attack}
\end{figure}
\subsection{Attack Primitives}

To satisfy Eq.~\ref{eq:game}, an adversary must perturb the dataset while ensuring its utility.
Motivated by software infringement analysis, which leverages both shallow (\eg, layout) and inherent (\eg, source code) features, we propose a top-down dataset feature hierarchy.
As shown in \Cref{fig:attack}, we distinguish shallow features at the sample and set levels from deep features at the distribution level, and categorize attack primitives accordingly into sample-, set-, and distribution-level primitives.

\begin{packeditemize}
    \item Attack \ding{182}: Post-processing Attack.
    The adversary post-processes (\eg, via Gaussian blurring) all samples in $\vicdst$ to prevent the auditor from identifying copied samples.

    \item Attack \ding{183}: Reorganization Attack.
    The adversary changes the dataset organization by combining a subset of $\vicdst$ with curated data following $\advdistriaccess$ to build $\advdst$.

    \item Attack \ding{184}: Polishing Attack.
    The adversary uses a generator (\ie, generative models trained on $\vicdst$ or public foundation models) to learn $\vicdistri$.
    The generator produces $\advdst$ from the same distribution as $\vicdst$ which has no identical sample but can train models of comparably high test accuracy~\cite{shumailov2024ai,wang2024do}.
\end{packeditemize}

We find that dataset regenerations can generally be decomposed into a combination of three attack primitives.
However, combining them (\eg, combining \ding{184} and \ding{182}) often compromises a dataset's utility (see \Cref{sec:evaluation}).
Therefore, instead of studying these complex combinations, we focus on individual primitives and vary their intensity to evaluate our tracking framework's robustness.
Recognizing that new primitives may appear in the future, we design our framework to be \textit{extensible} (see \Cref{sec:system}).

\begin{figure*}[t]
    \centering
    \includegraphics[width=0.95\linewidth]{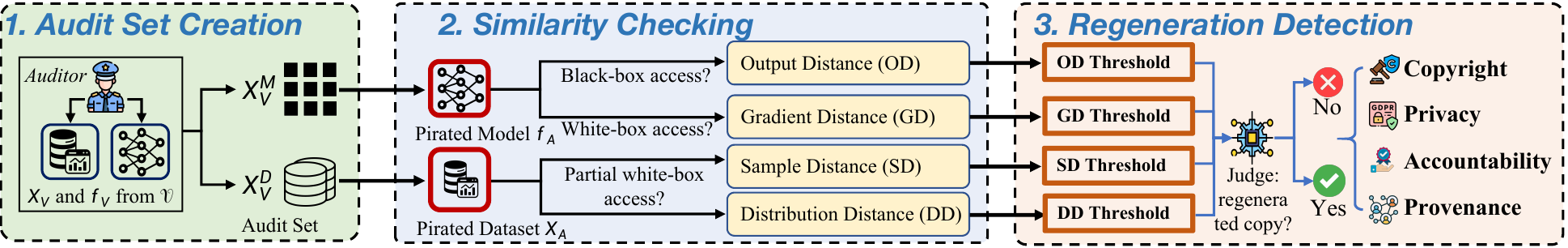}
    \caption{The workflow of \boxname. With the created audit set, the auditor measures the similarity scores based on their access to the target model/dataset.
    The final judgment is based on thresholding and applies to multiple downstream investigations.
    }
    \label{fig:overview}
\end{figure*}

\section{System Design of \boxname}
\label{sec:system}

In this section, we elaborate \boxname's design goals, defender assumptions, and workflow overview.

\subsection{Overview}

\subsubsection{Design Goals}
\boxname is tailored for tracking datasets from exclusive or scarce sources because of their proprietary value and the low likelihood of a coincidental match, which reflects a broader industry trend~\cite{NEURIPS2024_c3738949}.
\boxname can serve as a technical tool for auditors to support downstream legal assessment.
Moreover, \boxname can automatically screen uploaded datasets for potential regeneration attempts, analogous to automated plagiarism detection~\cite{arxiv_moderation} on preprint platforms. 
To sum up, the design goals include:
\begin{packeditemize}
    \item \textbf{Non-intrusiveness}: The tracking should not damage the source dataset \textit{utility}.
    \item \textbf{Robustness}: The regeneration threat should be accurately identified as long as the model-training utility is preserved.
    
    \item \textbf{Efficiency}: The tracking algorithm is of low complexity and only requires a small proportion of audited datasets.
    
    \item \textbf{Extensibility}: The framework should easily incorporate newer tracking metrics and should work under a variety of access levels to $\advf$ and $\advdst$.

\end{packeditemize}

\subsubsection{Defender Assumptions}

To support varying access settings, we consider the weakest defense assumption (\ie, black-box access to $\advf$ and no access to $\advdst$) and provide more accurate tracking as the defender capacity improves (\ie, white-box access to $\advf$ and a subset of $\advdst$).
We clarify the defender capacity as follows:
\begin{packedenumerate}
    \item \textit{Access to the suspect model $\advf$}.
    We consider two possibilities: 1) black-box access.
    The defender can only query $f_\cA$ to obtain confidence scores.
    2) white-box access.
    The defender (\eg, platform) can acquire $\advf$ weights.

    \item \textit{Access to the suspect dataset $\advdst$}.
    Similarly, the defender has 1) no access to $\advdst$'s samples, and 2) white-box access to the suspect sub-dataset $\advdst^s$, which is assumed to be uniformly sampled from $\advdst$.

\end{packedenumerate}

\begin{algorithm}[t]
\small
    \caption{$\boxname(\vicf, \vicdst, \advf, \advdst^s, N_{audit}^M, N_{audit}^D, seed, \cM_\cA, R_\cJ)$}
    \label{alg:overview}
    \KwIn{Victim's model $\vicf$ and dataset $\vicdst$, suspect model $\advf$ and subset $\advdst^s$, audit set sizes $N_{audit}^M$ and $N_{audit}^D$, random seed $seed$, access mode $\cM_\cA$, and judge requirement $R_\cJ$}
    \KwOut{Judgment $\cJ$}

    \SetKwFunction{GenAuditSet}{GenerateAuditSet}
    \SetKwFunction{Register}{Register}
    \SetKwFunction{Similarity}{SimilarityChecking}
    \SetKwFunction{Judging}{Judging}

    \tcp{Audit set generation (\Cref{subsec:judge_set_gen})}
    $\auditset^M$, $\auditset^D$ $\leftarrow$ \GenAuditSet$(\vicf, \vicdst, N_{audit}^M, N_{audit}^D, seed)$

    \tcp{Similarity metrics (\Cref{subsec:sim_check})}

    $\mathbf{s}_\cJ$ $\leftarrow$ \Similarity$(\auditset^M, \auditset^D, \vicf, (\advf, \advdst^s), \cM_\cA)$ 

    \tcp{Judge derived dataset (\Cref{subsec:judgement})}
    
    $\cJ$ $\leftarrow$ \Judging$(\mathbf{s}_\cJ, R_\cJ)$

    \Return $\cJ$
\end{algorithm}

Note that white-box access to a subset $\advdst^s$ is realistic.
For instance, current open-source or commercial platforms often provide a free preview subset (\eg, D0 in \Cref{tab:measure-dataset}).
Moreover, platforms or other parties responsible for copyright and license compliance~\cite{googleplay_copyright,vpvet_ccs2024}, can, subject to sharer authorization or regulatory mandates, leverage privacy-preserving computing techniques (\eg, Trusted Execution Environments (TEEs)~\cite{sgx_privacy_2023}) to obtain partial or complete access to the dataset.

\subsubsection{Workflow of \boxname.}

As \Cref{fig:overview} shows, \boxname consists of three modules and four similarity metrics spanning sample-, set-, and distribution-level features.
\Cref{alg:overview} depicts the procedure of \boxname.

The \textit{Audit Set Creation} module generates two audit sets $\auditset^M$ and $\auditset^D$ that are used by model-access and dataset-access metrics, respectively (Line 1).
To ensure efficiency, the audit sets should be small-sized and represent the source dataset.

To ensure extensibility, the \textit{Similarity Calculation} module contains multiple metrics (Line 2) that depend on the defender's access to $(\advf, \advdst)$, denoted by access mode $\cM_\cA$.
For each metric supported under $\cM_\cA$, the defender computes a metric value (a distance; smaller indicates higher similarity) and then thresholds it in the next module.

Finally, the \textit{Regeneration Judgment} module accomplishes the judgment based on the comprehensive evaluation of requirements $R_\cJ$, which will be elaborated in \Cref{subsec:judgement}.
The decision is accompanied by per-metric scores, which help interpret which feature scale drives the decision.

\begin{figure}[t]
    \centering
    \includegraphics[width=0.9\linewidth]{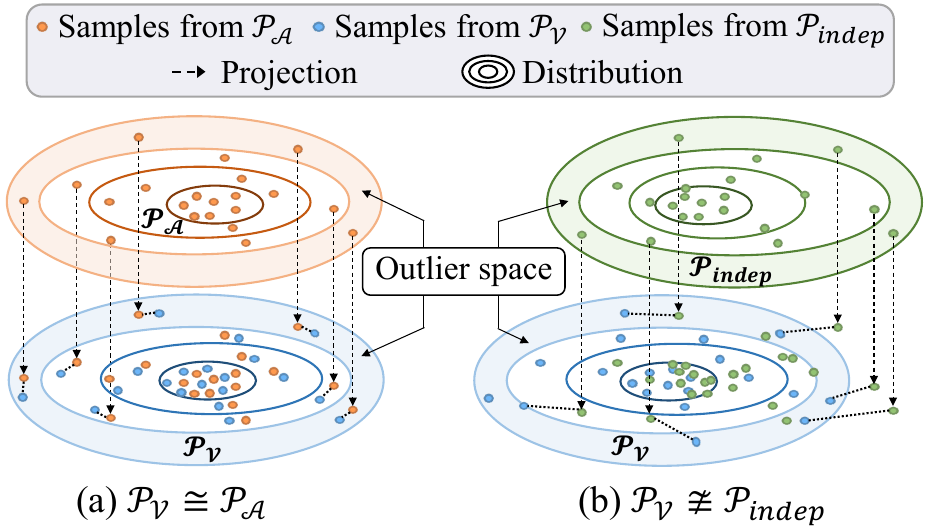}
    \caption{Intuition of our audit set generation.
    The distribution of a regenerated dataset is close to the victim's distribution ($\vicdistri\cong\advdistri$).
    On the other hand, for an independent dataset, the underlying distribution $\cP_{indep}$ has larger distance to the victim's distribution ($\vicdistri\ncong \cP_{indep}$).
    We select the outlier points (\ie, points on the shaded area) for further testing as they characterize the distribution contour.}
    \label{fig:illustration}
\end{figure}

\subsection{Audit Set Creation}
\label{subsec:judge_set_gen}

The defender first generates an audit set from $\vicdst$.
The audit set should be small to reduce computing cost and to well characterize the source dataset $\vicdst$ so that the derived datasets can be more accurately detected and independent datasets are less likely to be mistaken as a copy.
We consider two audit sets $\auditset^D$ and $\auditset^M$ used to calculate similarities for dataset metrics and model metrics, respectively.
The audit-set sizes are denoted by $N_{audit}^D$ and $N_{audit}^M$ for $\auditset^D$ and $\auditset^M$, respectively.
In \Cref{alg:overview}, we implement these steps as $\texttt{GenerateAuditSet}(\vicf,\vicdst,N_{audit}^M,N_{audit}^D,seed)$ that outputs $(\auditset^M,\auditset^D)$.
When the defender has access to the suspect sub-dataset $\advdst^s$, we uniformly sample from the victim dataset, \ie, $\auditset^D=\texttt{Uniform}(\vicdst, N_{audit}^D, seed)$, to better represent the distribution of $\vicdst$, where \texttt{Uniform} denotes a uniform sampler seeded with $seed$.

Considering the case with black-box or white-box access to the suspect model $\advf$, our insight is leveraging outliers to enlarge the difference between regenerated and independent datasets.
As illustrated in the \Cref{fig:illustration} (a), such outliers depict the dataset shape in the distribution space.
Hence, the regeneration effect on models can be amplified on the outliers, leading to more accurate detection.

For an independent dataset which follows a different $\cP_{indep}\ncong\vicdistri$, as shown in \Cref{fig:illustration} (b), the peripheral points can better characterize the distribution shift than the in-distribution points which are basically easy-to-fit samples.
Since the outliers are generally hard to fit, we sample $N_{audit}^M$ high-loss points as our audit set:
\begin{equation}
\begin{aligned}
    \auditset^M &= \texttt{Sample}(\{x|x\in\vicdst  \land  l(f_\cV(x)) \geq \tau_{audit}\}, N_{audit}^M),
\end{aligned}
\end{equation}
where $l$ is the loss function, $\tau_{audit}$ is the threshold of outlier range.

\subsection{Similarity Calculation}
\label{subsec:sim_check}

In this section, we present the four similarity metrics under different access modes.
As illustrated in \Cref{fig:overview}, the metrics are divided into two categories: metrics requiring model access and metrics requiring dataset access.
\Cref{tab:testing_metrics} presents the defender access and target for each metric.
If the defender does not have the required access for a metric, then \boxname skips this metric and proceeds with the rest.
Note that metrics designed for black-box model access can also be tested with white-box model access, and all metrics can be evaluated when the auditor has white-box access to the suspect model and partial white-box access to the suspect dataset.
The final judgment (\Cref{subsec:judgement}) is based on the holistic evaluation of metrics.

\begin{table}[t]
\centering
\caption{Summary of similarity metrics.}
\label{tab:testing_metrics}
\resizebox{0.9\linewidth}{!}{
\begin{tabular}{cccc}
\toprule

\multirow{2}{*}{\textbf{\begin{tabular}[c]{@{}c@{}}Similarity \\ Metric\end{tabular}}} &
  \multirow{2}{*}{\textbf{\begin{tabular}[c]{@{}c@{}}Defense \\ Target\end{tabular}}}&
  \multicolumn{2}{c}{\textbf{Defender Access}} \\ \cline{3-4} 
                  &            & $\cX_\cA$     & $f_\cA$     \\ \midrule
Output Distance (OD) &
  \multirow{3}{*}{\begin{tabular}[c]{@{}c@{}}Sample \& set features\end{tabular}} &
  / &
  \CIRCLE \\ \cline{1-1} \cline{3-4} 
Gradient Distance (GD) &            & / & \Circle \\ \cline{1-1} \cline{3-4} 
Sample Distance (SD)   &            & \LEFTcircle & /      \\ \hline
Distribution Distance (DD)      & Distribution features & \LEFTcircle & /      \\ \bottomrule
\end{tabular}
}
\begin{tablenotes}
\footnotesize
\item[] \CIRCLE: black-box; \Circle: white-box; \LEFTcircle: partial white-box; /: not involved.
\end{tablenotes}
\end{table}

\subsubsection{Metrics with Model Access}
Inspired by previous works~\cite{dataset_auditing_rl,maini2021dataset}, we found model memorization on training samples can reliably assess shallow dataset features.
Based on the access to the suspect model $\advf$, we propose two \textit{model-agnostic} metrics, Output Distance (OD), adapting from prior work~\cite{DBLP:conf/ndss/DongLCXZ023,dataset_auditing_rl} and our novel metric Gradient Distance (GD), respectively.

\bheading{Output Distance.}
The OD measures the distance of outputs on the audit set $\auditset^M$ between the victim's model $\vicf$ and the adversary's model $\advf$.
Since the audit set typically comprises hard-to-learn samples, a lower value indicates a higher likelihood that the adversarial model $\advf$ has also been trained on these audit samples.
We use $l_p$ norm to measure the distance between outputs:
\begin{equation}
    OD(\vicf, \advf, \auditset^M) = \frac{1}{\abs{\auditset^M}}\sum_{\mathbf{x}\in\auditset^M}\norm{\vicf(\mathbf{x})-\advf(\mathbf{x})}_p,
\end{equation}
where $\norm{\cdot}_p$ means the $l_p$ norm.
There are also other similar metrics such as Jensen-Shannon distance~\cite{DBLP:conf/isit/FugledeT04} in place of $l_p$ norm, but they are shown to have equivalent performance as OD~\cite{DBLP:conf/sp/ChenWPS0JM0S22}, thus we consider only $l_p$ and implement other metrics in future work.

\bheading{Gradient Distance.}
We propose GD for more detailed analysis with white-box access to $\advf$.
Empirically, for well-trained models, training data often exhibit lower loss, which correlates with smaller per-example gradients~\cite{maini2021dataset} under common losses (\Cref{fig:intuition_lng}).
We propose to compute the \textit{layer-wise normalized gradients}, $LNG(f, \mathbf{x})$, as a practical proxy signal, to quantify the self-influence of point $\mathbf{x}$ on model $f$:
\begin{equation}
    \mathrm{LNG}(f,x) \triangleq \frac{1}{N_f}\sum_{i=1}^{N_f}
\frac{g_i(f,x)}{\norm{\mathbf{w}_{f_i}}_p+\epsilon_{\text{num}}},
\end{equation}
where $N_f$ is the number of layers of $f$, $\mathbf{w}_{f_i}$ is the weight matrix of $i$-th layer and where $g_i(f,x) \triangleq \norm{\nabla_{\mathbf{w}_{f_i}}\,\ell\!\left(f(x),y\right)}_p$ and $\nabla_{\mathbf{w}_{f_i}}\ell(f(x),y)$ denotes the per-example gradient of the loss
with respect to the parameters of the $i$-th layer, evaluated at input $\mathbf{x}$.
Here, $y$ denotes the ground-truth label associated with $\mathbf{x}$.
The lower normalized gradient reflects the smaller impact of data $\mathbf{x}$ to the weight, and the layer-wise average represents the impact of $\mathbf{x}$ to the whole model.
Our novel metric GD is computed based on distance between $LNG$:

\begin{equation}
\resizebox{0.99\hsize}{!}{$
\begin{aligned}
    GD(\vicf, \advf, \auditset^M) = \frac{1}{\abs{\auditset^M}}\sum_{\mathbf{x}\in\auditset^M}|LNG(\vicf, \mathbf{x})
    -LNG(\advf, \mathbf{x})|.
\end{aligned}
$}
\end{equation}

\begin{figure}[t]
    \centering
    \includegraphics[width=0.95\linewidth]{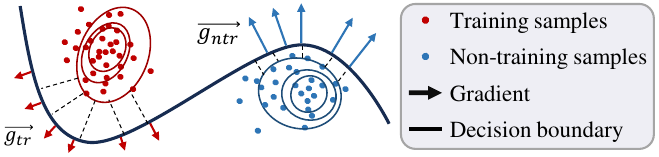}
    \caption{Intuition of GD: Training samples tend to induce smaller per-example gradients than non-training samples under converged training.}
    \label{fig:intuition_lng}
\end{figure}

\subsubsection{Metrics with Dataset Access}

With the suspect sub-dataset $\advdst^s$, the auditor can obtain more detailed testing results.
We incorporate two metrics in \boxname: Sample Distance (SD) and Distribution Distance (DD).

\bheading{Sample Distance.}
As demonstrated in 
\Cref{subsec:measure-method}, SD estimates the fraction of audit samples in $\auditset^D$ whose identical or similar copy also appears in $\advdst^s$, providing an approximation of actual proportion of duplicates.
As the adversary preserves the dataset quality, regenerated samples should have small distance to the original ones.
We use the closed $r$-ball $\cB_d(\mathbf{x}, r) = \{\mathbf{x}^\prime | d(\mathbf{x}^\prime, \mathbf{x})\leq r\}$ for each $\mathbf{x}\in\auditset^D$ to delineate the border of similar copies.
Formally:

\begin{equation}
    \resizebox{0.99\hsize}{!}{$SD(\auditset^D, \advdst^s, r) = \abs{\{\mathbf{x}|\mathbf{x}\in\auditset^D \land \cB_d(\mathbf{x},r)\cap\advdst^s \neq\emptyset \}}/\abs{\auditset^D}.$}
\end{equation}

\bheading{Distribution Distance.}
It is hard to directly measure distribution-level similarity, especially with only a small suspect subset.
To simultaneously satisfy the design goals, inspired by the notion of Maximum Mean Discrepancy (MMD)~\cite{DBLP:journals/jmlr/GrettonBRSS12,binkowski2018demystifying}, we propose to empirically estimate the distribution gap between $\auditset^D$ and $\advdst^s$ with:

\begin{equation}
\label{eq:dd}
\resizebox{0.99\hsize}{!}{$
    DD(\auditset^D, \advdst^s) = \expectedvalue_{\theta\sim \cP_\theta}\norm{\frac{1}{\abs{\auditset^D}} \sum_{i=1}^{\abs{\auditset^D}} \psi_\theta(\mathbf{x}_i) - \frac{1}{\abs{\advdst^s}} \sum_{j=1}^{\abs{\advdst^s}} \psi_\theta(\mathbf{x}_j^\prime)}_p,
    $}
\end{equation}
where $\psi_\theta$ is the feature extractor of parameter $\theta$, $\mathbf{x}_i\in\auditset^D$ and $\mathbf{x}_j^\prime\in\advdst^s$.
Appendix~\ref{appendix:additional_detail} presents the implementation details of $\psi_\theta$.

\subsection{Regeneration Judgment}
\label{subsec:judgement}

\begin{algorithm}[t]
\small
    \caption{$\texttt{Judging}(\mathbf{s}_\cJ, R_\cJ)$}
    \label{alg:judge}
    \KwIn{Metric values $\mathbf{s}_\cJ$ (smaller indicates higher similarity), and judge requirement $R_\cJ$}
    \KwOut{Judgment $\cJ$}
    \eIf{$\exists \lambda \in \{OD, GD\}  \text{ such that }  \mathbf{s}_\cJ[\lambda] \leq R_\cJ[\lambda]$}{
        \eIf{$SD$ is tested}{
        \Return Attack \ding{182} / \ding{183} with $\mathbf{s}_\cJ[SD]$.}{
            \Return Attack \ding{182} / \ding{183}.}
    }
    {
        \If{$\mathbf{s}_\cJ[DD] \leq R_\cJ[DD]$}{
            \Return Attack \ding{184}.}
    }

    \Return Independent.
\end{algorithm}

The judgment is based on the pre-computed threshold contained in the judge requirement $R_\cJ$ and the thresholding pattern among available metrics.
Inspired by one-sided hypothesis testing~\cite{DBLP:conf/sp/ChenWPS0JM0S22}, we calibrate each metric-specific threshold $\tau_\lambda$ using surrogate negative datasets $\{\cX_{neg}^i\}_i$ (independently constructed datasets).
For $\lambda\in\{OD,GD\}$, we train negative models $\{f_{neg}^i\}_i$ on $\{\cX_{neg}^i\}_i$ and compute $s_{neg}^i=\lambda(\vicf,f_{neg}^i,\auditset^M)$.
For $\lambda=DD$, we compute $s_{neg}^i=DD(\auditset^D,(\cX_{neg}^{i})^{s})$, where $(\cX_{neg}^{i})^{s}$ is a uniformly sampled audited subset of $\cX_{neg}^{i}$.
We then set $\tau_\lambda=\alpha_\lambda\cdot \min_i s_{neg}^i$, where $0<\alpha_\lambda\leq 1$ controls the sensitivity (larger $\alpha_\lambda$ is more permissive).
Since smaller distances indicate stronger similarity to $\vicdst$, using the closest negative score makes the threshold conservative for a small negative set.
The thresholding pattern consists of metrics whose values are lower than the corresponding thresholds (see \Cref{alg:judge}), serving as supporting evidence for the auditor to assess the regeneration attack implemented by the adversary.
When surrogate negative datasets are unavailable, the defender can alternatively set $\tau_\lambda=\alpha_\lambda\cdot \max_i s_{pos}^i$ based on potential pirated datasets $\{\cX_{pos}^i\}_i$, where $s_{pos}^i$ is computed analogously.
In this case, the threshold can also be determined via empirical strategies, \eg, simulating dataset modifications with surrogate datasets and selecting $\tau_\lambda$ to balance false positives and false negatives.
The computational cost mainly comes from surrogate model training, leading to a complexity of $O(N_f+N_g)$, where $N_f$ is the number of surrogate models and $N_g$ is the number of surrogate generators for simulating attacks.
We analyze threshold sensitivity in Appendix~\ref{app:fp_power}.

\subsection{Theoretical Analysis}
\label{sec:theory}
To provide a rigorous justification for \boxname, we analyze (i) why our similarity metrics are principled, and
(ii) the utility-evasion trade-off faced by a regeneration adversary. Detailed proofs are deferred to the \Cref{appendix:theory}.
We also analyze the final judgment as a statistical test in Appendix~\ref{app:fp_power}.

\subsubsection{Justification of Multi-Scale Similarity Metrics}
\label{sec:theory_metrics}
Our metrics capture distinct, fundamental properties of the learning pipeline.

OD is motivated by \emph{prediction stability}: under the utility constraint in Eq.~(\ref{eq:game}),
a successful regeneration should preserve generalization behavior, suggesting limited sensitivity to small data perturbations,
leading to close outputs on the audit set (Appendix~\ref{app:od}).
GD is motivated by \emph{influence-function} analysis: the self-influence of a sample is controlled by a quadratic form involving
its gradient and the inverse Hessian; GD serves as an efficient proxy for distinguishing seen-like samples from unseen-like ones
(Appendix~\ref{app:gd}).
DD estimates an MMD-inspired RKHS integral probability metric.
Under a characteristic kernel, the corresponding mean embedding is identifiable; in our implementation, we approximate this idea using finite random neural feature maps~\cite{DBLP:conf/icml/SaxeKCBSN11, DBLP:journals/tsp/GiryesSB16, DBLP:journals/corr/abs-2202-06438, DBLP:conf/wacv/ZhaoB23}, yielding a practical distribution-level similarity test (Appendix~\ref{app:dd}).

\subsubsection{Detectability via Domain Adaptation}
\label{sec:theory_limits}
The adversary faces a trade-off between evading detection and preserving utility (Eq.~\eqref{eq:game}).
We model this trade-off using domain adaptation theory.
Let $\advdistri$ denote the adversary's (regenerated) training distribution and $\vicdistri$ denote the victim's distribution on which utility is evaluated.
A standard generalization bound suggests that performing well on $\vicdistri$ depends on both (i) low risk on $\advdistri$
and (ii) small distributional divergence between $\advdistri$ and $\vicdistri$.

\begin{theorem}[Utility--Evasion Trade-off]
\label{thm:tradeoff}
Let $\mathcal{H}$ be a hypothesis class and consider the $0$--$1$ loss.
For any $h\in\mathcal{H}$,
\begin{equation}
R_{\vicdistri}(h)\;\le\; R_{\advdistri}(h)\;+\; \tfrac12\, d_{\mathcal{H}\Delta\mathcal{H}}(\advdistri,\vicdistri)\;+\;\lambda,
\label{eq:da_bound}
\end{equation}
where $\lambda=\min_{h\in\mathcal{H}}\big(R_{\advdistri}(h)+R_{\vicdistri}(h)\big)$  is the ideal joint error and $d_{\mathcal{H}\Delta\mathcal{H}}(\cdot,\cdot)$ is the $\mathcal{H}\Delta\mathcal{H}$-divergence.

\end{theorem}

Theorem~\ref{thm:tradeoff} highlights a generic trade-off between preserving utility on $\vicdistri$ and increasing
cross-domain discrepancy between $\advdistri$ and $\vicdistri$.
Our DD metric instantiates a practical, distribution-level
MMD-inspired discrepancy that is well-suited for detecting generative polishing, while the bound uses $d_{\mathcal{H}\Delta\mathcal{H}}$ to capture disagreement-based domain shift.
We treat these two notions as complementary rather than equivalent; empirically, increasing distribution-level deviation to evade DD tends to make maintaining high utility on $\vicdistri$ harder.
\begin{table*}[ht]
\centering
\caption{Details of benchmark datasets and experimental setup. For each task, we adopt one widely used training dataset as the victim's dataset $\vicdst$, and select three most similar public datasets used in the prior literature as the independent (negative) ones.}
\label{tab:dataset}
\resizebox{\linewidth}{!}{
\begin{tabular}{cccccccccc}
\toprule
\textbf{Task} & \textbf{\begin{tabular}[c]{@{}c@{}}Victim's \\ Dataset\end{tabular}} & \textbf{Domain} & \textbf{\begin{tabular}[c]{@{}c@{}}Input\\Size\end{tabular}} & \textbf{\# of Classes} & \textbf{\begin{tabular}[c]{@{}c@{}}Victim's \\ Architecture\end{tabular}} & \textbf{\begin{tabular}[c]{@{}c@{}}Additional \\ Architecture for $\advf$\end{tabular}}&\textbf{\begin{tabular}[c]{@{}c@{}}Post-processing in\\Attack \ding{182}\end{tabular}}                                                                    & \textbf{\begin{tabular}[c]{@{}c@{}} Generators for \\Attack \ding{184}\end{tabular}} & \textbf{\begin{tabular}[c]{@{}c@{}}Independent\\ Datasets\end{tabular}} \\ \midrule
\multicolumn{1}{c}{\begin{tabular}[c]{@{}l@{}}Handwritten Digit\\Classification\end{tabular}}
  & MNIST & Image & 32$\times$32 & 10 & LeNet-5 & 3-layer MLP & \multirow{3}{*}{\begin{tabular}[c]{@{}c@{}}Glass blur,\\ Impulse noise, \\ JPEG compression, \\Pixelate, \\Spatter, \\ Speckle noise\end{tabular}} & Conditional DCGAN & \begin{tabular}[c]{@{}c@{}}
  DIDA~\cite{Kusetogullari2020DIDA} (Neg-1), \\ ARDIS~\cite{DBLP:journals/nca/KusetogullariYC20}
   (Neg-2),\\ SVHN~\cite{netzer2011reading} (Neg-3)\end{tabular}   \\ \cline{1-7} \cline{9-10} 
\multicolumn{1}{c}{\begin{tabular}[c]{@{}l@{}}Object\\Classification\end{tabular}} & CIFAR-10 & Image & 32$\times$32 & 10 & ResNet-18~\cite{he2016deep} & \begin{tabular}[c]{@{}l@{}}MobileNet-v2~\cite{DBLP:conf/cvpr/SandlerHZZC18}, \\VGG-13~\cite{vgg}\end{tabular} & & EDM~\cite{DBLP:conf/nips/KarrasAAL22} & \begin{tabular}[c]{@{}c@{}}STL-10~\cite{coates2011analysis} (Neg-1), \\ ImageNet-Set1 (Neg-2), \\ ImageNet-Set2 (Neg-3) \end{tabular}\\ \cline{1-7} \cline{9-10} 
\multicolumn{1}{c}{\begin{tabular}[c]{@{}l@{}}Facial Attribute \\Classification\end{tabular}}& FairFace & Image & 128$\times$128 & 20 & ResNet-101 &  \multicolumn{1}{c}{\begin{tabular}[c]{@{}l@{}}EfficientNet-v2-s~\cite{DBLP:conf/icml/TanL21}, \\RegNet-y-8gf~\cite{regnet_architecture}\end{tabular}} & & StyleGAN3~\cite{DBLP:conf/nips/KarrasALHHLA21} & \begin{tabular}[c]{@{}c@{}} LFWA+~\cite{liu2015faceattributes} (Neg-1), \\PubFig~\cite{DBLP:conf/iccv/KumarBBN09} (Neg-2),\\UTKFace~\cite{utkface_dataset} (Neg-3)\end{tabular} \\ \hline
\multicolumn{1}{c}{\begin{tabular}[c]{@{}l@{}}News \\Classification\end{tabular}}  & AG-News & Text & 768 & 4 & Tiny-BERT~\cite{tinymini-bert} & \begin{tabular}[c]{@{}c@{}}Mini-BERT,\\Small-BERT~\cite{bhargava2021generalization}\end{tabular} & \begin{tabular}[c]{@{}c@{}}Word Addition, \\Word Deletion, \\ Synonym Swap\end{tabular} & \begin{tabular}[c]{@{}c@{}}ChatGPT-enhanced\\ T5 paraphraser~\cite{chatgpt_t5_paraphraser}\end{tabular} & \begin{tabular}[c]{@{}c@{}}Yahoo-News~\cite{DBLP:conf/emnlp/YangXWL19} (Neg-1),\\
CNN-News~\cite{cnn_dataset} (Neg-2),\\ Guardian News (Guard-News)~\cite{guardian_dataset} (Neg-3)\end{tabular} \\ \bottomrule
\end{tabular}}

\end{table*}

\begin{table*}[t]
\centering
\caption{Similarity measurement with our proposed metrics on regenerated datasets produced by the sample-level attack and on the independent datasets for vision tasks.}
\label{tab:a1mod_all_distance}
\resizebox{0.99\textwidth}{!}{%
\begin{tabular}{ccccccccccccccccc}
\toprule
\multicolumn{2}{c}{\multirow{2}{*}{\textbf{\begin{tabular}[c]{@{}c@{}}Dataset Type\end{tabular}}}}                          & \multicolumn{5}{c}{\textbf{MNIST}}                                                                                          & \multicolumn{5}{c}{\textbf{CIFAR-10}}                                                                                      & \multicolumn{5}{c}{\textbf{FairFace}}                                                                                         \\ \cline{3-7}\cline{8-17} 
\multicolumn{2}{c}{}                                                                                                               & \multicolumn{1}{c}{UR (\%)} & \multicolumn{1}{c}{OD}    & \multicolumn{1}{c}{GD}    & \multicolumn{1}{c}{SD}    & DD     & \multicolumn{1}{c}{UR (\%)} & \multicolumn{1}{c}{OD}    & \multicolumn{1}{c}{GD}    & \multicolumn{1}{c}{SD}    & DD    & \multicolumn{1}{c}{UR (\%)} & \multicolumn{1}{c}{OD}    & \multicolumn{1}{c}{GD}    & \multicolumn{1}{c}{SD}    & DD     \\ \midrule
\multicolumn{1}{c}{\multirow{6}{*}{\begin{tabular}[c]{@{}c@{}}Suspect\\ Dataset\\ (Positive)\end{tabular}}}     & Glass blur       & \multicolumn{1}{c}{73.77}   & \multicolumn{1}{c}{0.155} & \multicolumn{1}{c}{0.338} & \multicolumn{1}{c}{0.0}   & 7.988  & \multicolumn{1}{c}{90.17}   & \multicolumn{1}{c}{0.235} & \multicolumn{1}{c}{0.438} & \multicolumn{1}{c}{0.254} & 1.336 & \multicolumn{1}{c}{84.13}   & \multicolumn{1}{c}{0.032} & \multicolumn{1}{c}{0.158} & \multicolumn{1}{c}{0.946} & 2.485  \\ \cline{2-17} 
\multicolumn{1}{c}{}                                                                                            & Impulse noise    & \multicolumn{1}{c}{98.60}   & \multicolumn{1}{c}{0.066} & \multicolumn{1}{c}{0.148} & \multicolumn{1}{c}{0.0}   & 2.889  & \multicolumn{1}{c}{92.36}   & \multicolumn{1}{c}{0.197} & \multicolumn{1}{c}{0.421} & \multicolumn{1}{c}{0.0}   & 1.427 & \multicolumn{1}{c}{82.67}   & \multicolumn{1}{c}{0.075} & \multicolumn{1}{c}{0.374} & \multicolumn{1}{c}{0.0}   & 3.200  \\ \cline{2-17} 
\multicolumn{1}{c}{}                                                                                            & JPEG compression & \multicolumn{1}{c}{98.34}   & \multicolumn{1}{c}{0.041} & \multicolumn{1}{c}{0.122} & \multicolumn{1}{c}{1.0}   & 1.453  & \multicolumn{1}{c}{90.78}   & \multicolumn{1}{c}{0.209} & \multicolumn{1}{c}{0.383} & \multicolumn{1}{c}{0.988} & 1.262 & \multicolumn{1}{c}{85.25}   & \multicolumn{1}{c}{0.044} & \multicolumn{1}{c}{0.228} & \multicolumn{1}{c}{1.0}   & 2.348  \\ \cline{2-17} 
\multicolumn{1}{c}{}                                                                                            & Pixelate         & \multicolumn{1}{c}{98.28}   & \multicolumn{1}{c}{0.068} & \multicolumn{1}{c}{0.171} & \multicolumn{1}{c}{0.002} & 3.241  & \multicolumn{1}{c}{89.96}   & \multicolumn{1}{c}{0.296} & \multicolumn{1}{c}{0.691} & \multicolumn{1}{c}{0.988} & 1.269 & \multicolumn{1}{c}{87.46}   & \multicolumn{1}{c}{0.027} & \multicolumn{1}{c}{0.117} & \multicolumn{1}{c}{1.0}   & 2.393  \\ \cline{2-17} 
\multicolumn{1}{c}{}                                                                                            & Spatter          & \multicolumn{1}{c}{96.97}   & \multicolumn{1}{c}{0.070} & \multicolumn{1}{c}{0.155} & \multicolumn{1}{c}{0.008} & 3.479  & \multicolumn{1}{c}{92.58}   & \multicolumn{1}{c}{0.240} & \multicolumn{1}{c}{0.497} & \multicolumn{1}{c}{0.129} & 2.693 & \multicolumn{1}{c}{82.22}   & \multicolumn{1}{c}{0.056} & \multicolumn{1}{c}{0.284} & \multicolumn{1}{c}{0.006} & 4.272  \\ \cline{2-17} 
\multicolumn{1}{c}{}                                                                                            & Speckle noise    & \multicolumn{1}{c}{98.73}   & \multicolumn{1}{c}{0.033} & \multicolumn{1}{c}{0.106} & \multicolumn{1}{c}{0.526} & 2.374  & \multicolumn{1}{c}{91.08}   & \multicolumn{1}{c}{0.259} & \multicolumn{1}{c}{0.508} & \multicolumn{1}{c}{0.597} & 1.321 & \multicolumn{1}{c}{76.27}   & \multicolumn{1}{c}{0.061} & \multicolumn{1}{c}{0.305} & \multicolumn{1}{c}{0.068} & 2.768  \\ \hline
\multicolumn{1}{c}{\multirow{3}{*}{\begin{tabular}[c]{@{}c@{}}Independent\\ Dataset\\ (Negative)\end{tabular}}} & Neg-1            & \multicolumn{1}{c}{13.48}   & \multicolumn{1}{c}{1.254} & \multicolumn{1}{c}{4.073} & \multicolumn{1}{c}{0.0}   & 35.927 & \multicolumn{1}{c}{32.19}   & \multicolumn{1}{c}{1.053} & \multicolumn{1}{c}{5.951} & \multicolumn{1}{c}{0.0}   & 4.586 & \multicolumn{1}{c}{15.23}   & \multicolumn{1}{c}{1.122} & \multicolumn{1}{c}{6.256} & \multicolumn{1}{c}{0.000} & 9.441  \\ \cline{2-17} 
\multicolumn{1}{c}{}                                                                                            & Neg-2            & \multicolumn{1}{c}{18.10}   & \multicolumn{1}{c}{1.309} & \multicolumn{1}{c}{4.693} & \multicolumn{1}{c}{0.0}   & 33.181 & \multicolumn{1}{c}{77.16}   & \multicolumn{1}{c}{0.467} & \multicolumn{1}{c}{4.944} & \multicolumn{1}{c}{0.009} & 1.652 & \multicolumn{1}{c}{22.05}   & \multicolumn{1}{c}{1.098} & \multicolumn{1}{c}{6.018} & \multicolumn{1}{c}{0.000} & 13.258 \\ \cline{2-17} 
\multicolumn{1}{c}{}                                                                                            & Neg-3            & \multicolumn{1}{c}{63.78}   & \multicolumn{1}{c}{0.339} & \multicolumn{1}{c}{0.901} & \multicolumn{1}{c}{0.0}   & 26.737 & \multicolumn{1}{c}{82.74}   & \multicolumn{1}{c}{0.451} & \multicolumn{1}{c}{4.860} & \multicolumn{1}{c}{0.008} & 1.870 & \multicolumn{1}{c}{58.60}   & \multicolumn{1}{c}{0.995} & \multicolumn{1}{c}{6.598} & \multicolumn{1}{c}{0.001} & 13.589 \\ \bottomrule
\end{tabular}
}

\end{table*}

\begin{table}[t]
\centering
\caption{Similarity measurement with our metrics in the news classification task.}
\label{tab:a1mod_all_distance_agnews}
\resizebox{\linewidth}{!}{%
\begin{tabular}{ccccccc}
\toprule
\multicolumn{2}{c}{\textbf{Modification Type}}                                                                              & UR (\%) & OD    & GD     & SD    & DD     \\ \midrule
\multicolumn{1}{c}{\multirow{3}{*}{\begin{tabular}[c]{@{}c@{}}Suspect\\ Dataset\\ (Positive)\end{tabular}}}     & Insertion & 99.98   & 0.106 & 2.967  & 0.002 & 4.789  \\ \cline{2-7} 
\multicolumn{1}{c}{}                                                                                            & Deletion  & 99.31   & 0.137 & 5.123  & 0.084   & 5.952  \\ \cline{2-7} 
\multicolumn{1}{c}{}                                                                                            & Synonym   & 99.48   & 0.160 & 3.493  & 0.002   & 7.104  \\ \hline
\multicolumn{1}{c}{\multirow{3}{*}{\begin{tabular}[c]{@{}c@{}}Independent\\ Dataset\\ (Negative)\end{tabular}}} & Neg-1     & 47.95  & 1.070 & 15.786 & 0.0   & 11.294 \\ \cline{2-7} 
\multicolumn{1}{c}{}                                                                                            & Neg-2     & 60.83  & 1.069 & 19.709 & 0.0   & 13.913 \\ \cline{2-7} 
\multicolumn{1}{c}{}                                                                                            & Neg-3     & 80.78  & 0.600 & 15.157 & 0.0   & 13.165 \\ \bottomrule
\end{tabular}
}

\end{table}

\section{Evaluation}
\label{sec:evaluation}

In this section, we first present the experimental settings in \Cref{subsec:expsetup}.
Then, in \Cref{subsec:exp_defend_attack1,subsec:exp_defend_attack2,subsec:exp_defend_attack3}, we evaluate the effectiveness of \boxname under different defender access assumptions against Attack~\ding{182}, Attack~\ding{183}, and Attack~\ding{184}, respectively.
We report OD, GD, and DD values to represent the defender's black-box model access, white-box model access, and white-box access to a partial suspect dataset, respectively.
Next, we compare the True Positive Rate (TPR) of \boxname with previous baselines in \Cref{subsec:exp_compare} and analyze potential false positives in Appendix~\ref{app:fp_power}.
Finally, we study adaptive attacks in \Cref{subsec:exp_robust_potential_attack}.

\subsection{Datasets \& Setup}
\label{subsec:expsetup}
We use widely adopted AI training datasets (\eg, CIFAR-10) for our evaluation.
For each task, the victim holds a training dataset as $\vicdst$ and trains models achieving high accuracy on corresponding test data (simulated as real-world data of distribution $\cP_t$).
The adversary only has access to independent (negative) datasets for the same task, whose distributions are similar but not identical to $\vicdst$.
For text data, we use BERT \texttt{[CLS]} embeddings so that inputs are represented as vectors, analogous to image features.

\bheading{Regenerated Datasets.}
We craft 320 regenerated datasets in total.
For the Attack \ding{182}, we first test the 6 post-processing attacks as listed in \Cref{tab:dataset} (see \Cref{fig:a1_visualize} \Cref{fig:example_agnews} for examples), and fix the post-processing severity to the low level.
We then increase the severity to test the robustness in \Cref{subsec:exp_defend_attack1} and extend to 20 different post-processing attack variants in \Cref{tab:comparison}.
For the set-level attack (Attack \ding{183}), we consider the adversary randomly mixing a subset of the victim dataset $\vicdst$ with one independent dataset to form $\advdst$ with overlap ratio $p_\cV\in\{0.2, 0.4, 0.6, 0.8\}$ where $p_\cV=\abs{\vicdst\cap\advdst} / \abs{\vicdst}$.
For the distribution-level attack (Attack \ding{184}), we consider both generative adversarial networks (GANs) and diffusion models in our experiments as they represent common generation strategies.
We use the conditional GAN for MNIST and one of the state-of-the-art diffusion models EDM~\cite{DBLP:conf/nips/KarrasAAL22} for CIFAR-10.
For FairFace, we train the GAN model StyleGAN3~\cite{DBLP:conf/nips/KarrasALHHLA21}.
For text data, we use a ChatGPT-enhanced T5 paraphraser~\cite{chatgpt_t5_paraphraser}.

\bheading{Negative Datasets.}
Prior work~\cite{maini2021dataset,huang2022dataset,DBLP:conf/ndss/DongLCXZ023,datause_ccs2024} evaluates on negative datasets from the same distribution as the victim dataset, which, although valid for tracking sample use, contradicts our attack assumption: an adversary does not need to regenerate the dataset if they already have access to data from the same distribution.

To the best of our knowledge, there is \emph{no} public benchmark to evaluate dataset regeneration.
Therefore, we build a benchmark with research-permissive open-source datasets.
We select negative datasets that are different from the victim dataset but address the same task to induce distributional discrepancy; however, \emph{few} public datasets provide an exact class match.
To ensure broad coverage, we surveyed 26 public datasets reported in prior work and, after manually harmonizing labels, selected the three most similar datasets as negatives (see \Cref{tab:dataset}).

Note that the rarity of negative datasets with similar data distributions supports our assumption that attackers typically cannot obtain a dataset with the same distribution.
Nevertheless, there are honest but coincident cases.
For CIFAR-10, due to the limited choices of negative datasets, we sample negative datasets from ImageNet, which shares a common source (WordNet) with CIFAR-10.

\bheading{Audit Set.}
In terms of the audit set $\auditset^M$ (\Cref{subsec:judge_set_gen}), we set $\tau_{audit}=0.2$ and $N_{audit}^M=1000$, and sample $\auditset^M$ uniformly at random from the high-loss set induced by $\tau_{audit}$.
To compute SD and DD, we set $N_{audit}^D = 0.01\cdot|\vicdst|$ and sample $\auditset^D$ uniformly from $\vicdst$, and we set $|\advdst^s| = 0.01\cdot|\advdst|$ by uniformly sampling a fixed subset from $\advdst$; both subsets remain unchanged throughout the experiments.
We approximate the expectation in DD (Eq.~\ref{eq:dd}) by averaging over 1000 random draws of $\psi_\theta$.
Each draw corresponds to sampling the parameters $\theta$ of a two-layer MLP from $\mathcal{N}(0,1)$ as a universal feature map $\psi_\theta$~\cite{DBLP:conf/icml/SaxeKCBSN11, DBLP:journals/tsp/GiryesSB16, DBLP:journals/corr/abs-2202-06438, DBLP:conf/wacv/ZhaoB23} (Appendix~\ref{appendix:additional_detail}).

\bheading{Models.}
In total, we trained 590 classification and generative models.
We choose one widely studied architecture for the victim model $\vicf$ for each task.
In addition to the same architecture as the victim, the adversary chooses more advanced architectures as listed in \Cref{tab:dataset} to make up for the dataset utility degradation.
The default training setting is in \Cref{appendix:additional_detail}.
We further evaluate \boxname's robustness against different training strategies by altering training epochs, architecture, and attack severity.

\bheading{Utility Metrics.}
We quantify the adversary's gain by the Utility Retention (UR), reported as a percentage: $UR = 100\times v(\advf) / v(\vicf)$, where $v$ is the test accuracy.
Take CIFAR-10 as an example, $v(\vicf)$ represents the test accuracy on 10,000 test samples of model $\vicf$ trained on 50,000 training samples.
A UR closer to 100\% indicates greater utility retention.

\subsection{Defending Sample-level Attack}
\label{subsec:exp_defend_attack1}
We first evaluate the effectiveness of \boxname for the sample-level attack (Attack \ding{182}) with the default training setting.
Then, we investigate the robustness against various training settings of $\advf$ and the post-processing severity. 
We mainly report OD, GD and DD results for the defender with access to the black-box suspect model, white-box suspect model and partial suspect dataset.

\bheading{Effectiveness.}
In \Cref{tab:a1mod_all_distance}, we present the metric values on datasets regenerated by Attack \ding{182} and on the negative datasets for the vision data.
\Cref{tab:a1mod_all_distance_agnews} shows corresponding results for the text data, where the text embeddings produced by $\vicf$ are compared. 
We also provide the UR to assess utility loss during regeneration.

\iheading{Judging Results.}
Under low attack severity, \boxname can consistently flag datasets regenerated under the sample-level attack (Attack~\ding{182}) as shallow-feature regeneration (Attack~\ding{182}/\ding{183}) following \Cref{alg:judge}.
We observe that OD, GD and DD have lower values for datasets regenerated by our tested post-processing approaches.
Besides, their metric values on negative datasets are all higher than on the positive cases and the determined threshold.
Consequently, the auditor can identify these cases as shallow-feature regeneration following \Cref{alg:judge} under any defender access mode in \Cref{tab:testing_metrics}.

\iheading{Unequal Effect of Post-processing Methods.}
We note that the same post-processing approach can produce unequal impact across different datasets or models.
For example, MNIST data are more sensitive to the ``Glass blur'' than the other two vision tasks (CIFAR-10 and FairFace) because of larger accuracy degradation.
In addition, facial attribute data are more sensitive across post-processing methods.

\iheading{Utility Gap of Negative Datasets.}
The negative datasets have lower UR than datasets regenerated under Attack~\ding{182} because of the inherent data distribution discrepancy.
For example, the negative datasets of the facial attribute classification have less balanced data, which hinders the model from learning facial features equally among races and genders.
This also implies that, in case of exclusive data source, a dataset of almost equal model-training utility to the proprietary dataset is more suspicious for unauthorized copying under the exclusive-source assumption.

\iheading{Results on Texts.}
We have similar observations for text data (\Cref{tab:a1mod_all_distance_agnews}): the models trained on datasets that undergo random insertion, synonym swap, and deletion for 30\% words can achieve near-perfect UR on test data.
We note that the URs of derived datasets are much higher and close to 100\%.
This signifies that the text embedding models are not fully influenced by individual words but also the semantic meaning of input texts.

\iheading{Evaluation of SD.}
It is worth noting that SD metric generates unstable values for the sample-level regeneration attack.
Ideally, nearly all samples are copied from the original dataset with additional modification, thus SD should be close to 1. 
However, the value of SD varies from 0 to 1.
For instance, the SD is significantly lower than other post-processing methods for CIFAR-10 post-processed by ``Impulse noise'' and is nearly 0 for the facial dataset post-processed by ``Spatter''.
The reason is that we choose a conservative radius for SD and a single choice of radius cannot fully characterize closeness between samples.
Higher radius should be better for capturing the post-processed samples but can cause more false positives.
Another potential reason is that we adopt classic $L_p$ norm to compare sample closeness, which is computationally efficient but does not capture the visual similarity.  
We will investigate the more robust metrics (\eg, LPIPS~\cite{DBLP:conf/cvpr/ZhangIESW18}) in the future work.

\bheading{Robustness.}
\Cref{fig:a1_rosbustness} also validates the robustness of our framework when the adversary adopts different settings to train their model $\advf$, and post-processing methods of medium and high severity.
For OD and GD, we investigate the robustness against the adversary's different training settings.
As expected, more training epochs (More Epochs) or more advanced model architecture (Diff. Archi.) result in better generalization over the modified samples and thus higher UR; however, these settings do not enable evasion and the regenerated datasets remain detectable under our previously determined thresholds.
Besides, we also investigate when the adversary increases the modification severity (Higher Sev.): amplifying the post-processing to medium and high level.

\begin{figure}[t]
    \centering
    \includegraphics[width=0.90\linewidth]{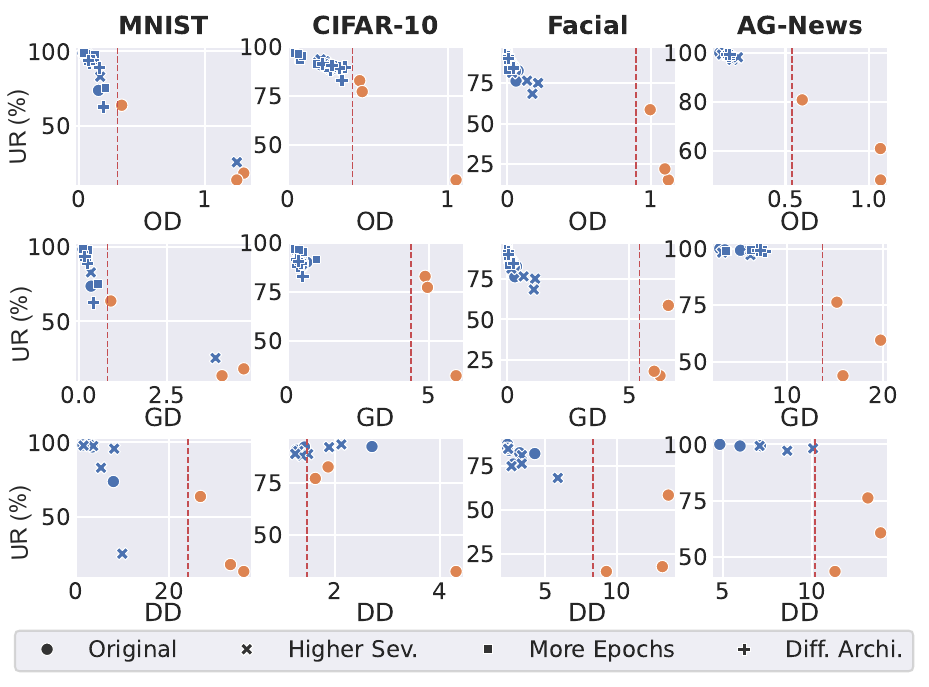}
    \caption{Robustness evaluation for defending the Attack~\ding{182}.
    We show the trade-off between UR and distances under various attack settings, where each point represents a model. The vertical dashed line is the threshold.}
    \label{fig:a1_rosbustness}
\end{figure}

\iheading{Trade-off between UR and Similarity.}
From the results of vision (\Cref{tab:a1mod_all_distance}) and text data (\Cref{tab:a1mod_all_distance_agnews}), we notice the trade-off between UR and the value of OD, GD and DD: dissimilar datasets result in lower UR and higher metric values, and datasets regenerated by the sample-level attack have higher UR and lower metric values.
Note that the trade-off is not linear: some negative datasets can achieve simultaneously higher UR and lower testing metric value than another negative dataset.
However, they are still separable from the regenerated datasets. 

In \Cref{fig:a1_rosbustness}, we visualize the trade-off under more complicated training and post-processing settings. 
The red vertical line represents the decision threshold determined previously. 
We observe that there are clear intervals between positive and negative datasets for OD and GD, which ensures the accurate judgment by \Cref{alg:judge}.
As we consider the sample-level attack, the inherent data distribution can be shifted from the original distribution, leading to a DD value higher than the decision threshold.
In summary, we observe the trade-off between the UR and our metrics: the adversary cannot simultaneously achieve high UR and evade our tracking by \Cref{alg:judge} under different attack settings.

\iheading{Case Study: CIFAR-10 and ImageNet.}
In \Cref{fig:a1_rosbustness}, we observe two negative datasets for CIFAR-10 are closer to the positive datasets than other tasks.
After manual checking, they are sampled from ImageNet.
The reason is both CIFAR-10 and ImageNet are sampled from the same source (\ie, searching on Internet using hierarchical structure of WordNet during close periods).
Therefore, they are more similar than other tasks.

\begin{tcolorbox}[notitle, boxrule=0.5pt,left=0.05cm, right=0.05cm, top=0.05cm, bottom=0.05cm]
\textbf{Takeaway 1:}
\boxname can effectively detect datasets regenerated under Attack~\ding{182} on both vision and text data and is robust against various sample-level attack settings.
\end{tcolorbox}

\subsection{Defending Set-level Attack}\label{subsec:exp_defend_attack2}
In this section, we first study the effectiveness against set-level regeneration under overlap ratios $p_\cV\in\{0.2, 0.4, 0.6, 0.8\}$, then show robustness under different training strategies for $f_\cA$.

\begin{figure}[t]
    \centering
    \includegraphics[width=0.99\linewidth]{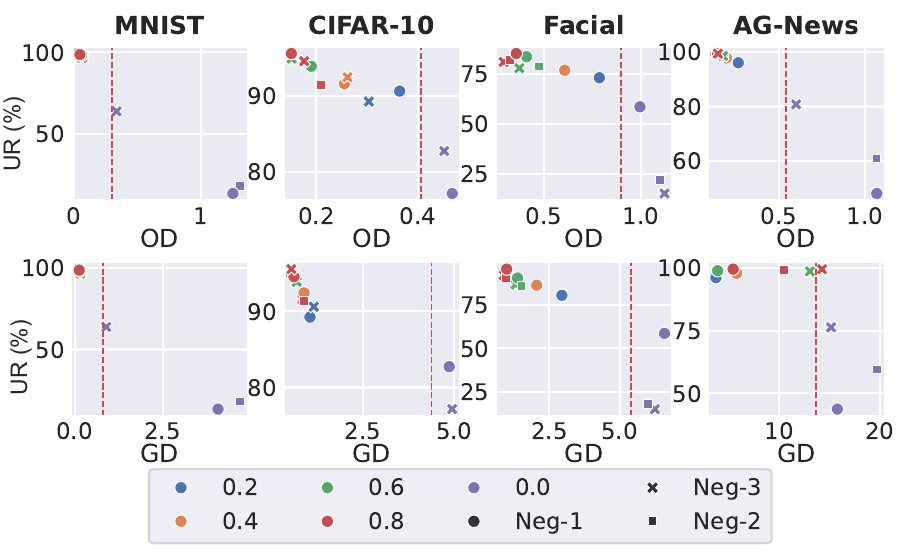}
    \caption{
    Trade-off between the measured distances (OD and GD) and the UR for intersection ratios $0.2, 0.4, 0.6, 0.8$ and the negative ones (\ie, ratio $0.0$), where each point represents a model.
    The vertical dashed line is the threshold.
    }
    \label{fig:a2_trade_off_odgd}
\end{figure}

\bheading{Effectiveness.}
We investigate the model-level metrics (OD and GD) in \Cref{fig:a2_trade_off_odgd} under defender's black-box and white-box access to the suspect model, and the dataset-level metrics (SD and DD) in \Cref{fig:a2_sddd} when the defender can access partial suspect dataset.

\iheading{Evaluation of OD and GD.}
\Cref{fig:a2_trade_off_odgd} presents the UR-metric trade-off for OD and GD on our tested severity.
We show the test intersection ratios in different colors and use the point style (\eg, circle or square) to represent the mixed independent datasets.
As a higher similarity means more mixed instances in the suspect dataset, the OD and GD have smaller value as the intersection ratio increases due to sample-level closeness.
Note that we do not aim to estimate the exact similarity but only to judge whether the suspect dataset contains a significant proportion of victim samples.
Additionally, we note that the GDs measured on models finetuned on AG-News are close to or higher than the threshold.
We provide the in-depth analysis along with the robustness.
However, we observe that in general the trade-off between UR and our metrics (OD and GD) holds, and that the adversary cannot craft a mixed dataset under Attack~\ding{183} without sacrificing the utility.

\begin{figure}[t]
    \centering
    \includegraphics[width=0.95\linewidth]{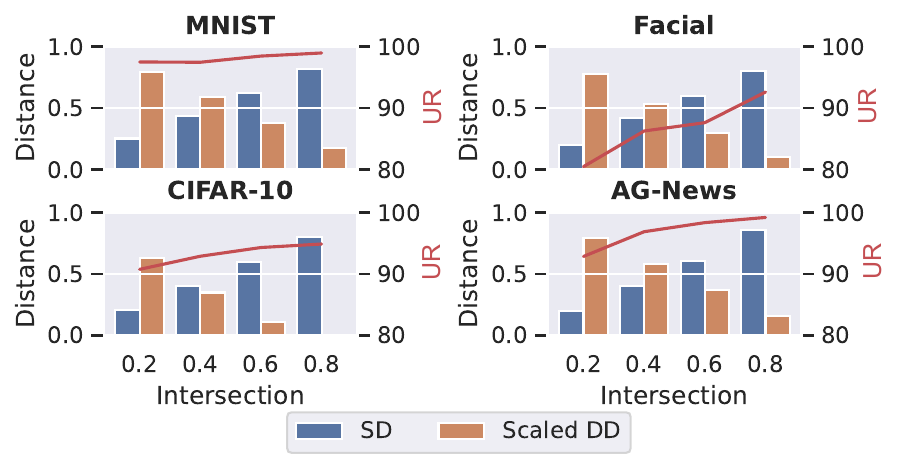}
    \caption{Trade-off between dataset metrics (SD and DD) and the UR for the intersection ratios $0.2, 0.4, 0.6, 0.8$ and the negative one $0.0$. The red curve is the UR (right axis).}
    \label{fig:a2_sddd}
\end{figure}

\begin{figure}[t]
    \centering
    \includegraphics[width=0.9\linewidth]{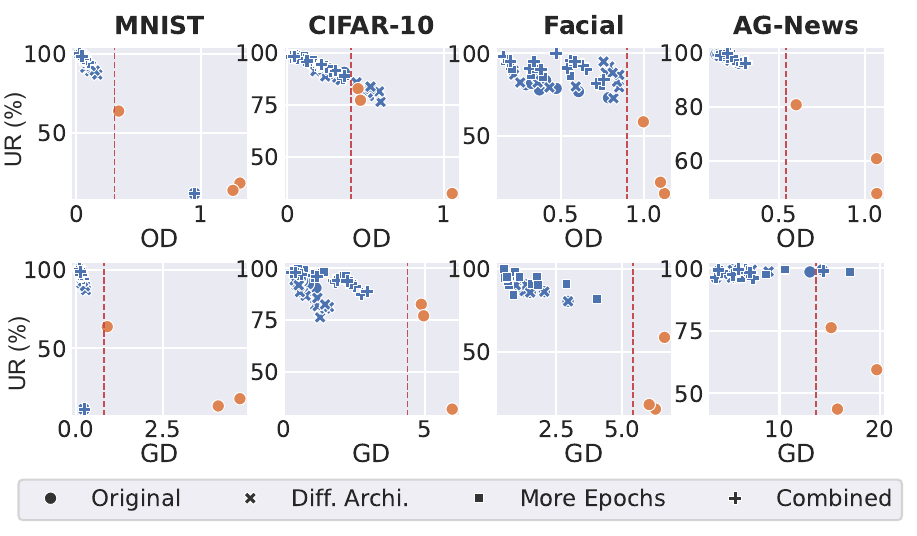}
    \caption{
    Robustness evaluation of defending Attack~\ding{183}.
    The trade-off between the UR and our metrics under different attack severity and adversary model training settings.
    The vertical dashed line represents the threshold.
    }
    \label{fig:a2_rosbustness}
\end{figure}

\iheading{Evaluation of SD and DD.}
We plot the results in \Cref{fig:a2_sddd}.
These two metrics are evaluated with access to the audit set and suspect sub-dataset, enabling a more direct dataset-level comparison.
Here, SD estimates sample-level overlap based on the audit set, while DD captures distribution-level distance, which decreases as the ground-truth overlap ratio increases.

\bheading{Robustness.}
\Cref{fig:a2_rosbustness} evaluates the robustness under different training settings for the adversary's model.
\boxname can correctly detect datasets regenerated under Attack~\ding{183} setting because the GD metric can accurately distinguish the positive and negative datasets on the vision task.
Note that there are tested positive models with OD values higher than the threshold for CIFAR-10, because ImageNet, from which the two negative datasets are sampled with manual annotation alignment, is similar to CIFAR-10 and often selected as the source for CIFAR-10 alternatives (\eg, CINIC~\cite{darlow2018cinic}).
On the other hand, the STL-10 is more different from CIFAR-10 (\eg, in terms of the image tone) and is accurately detected.

For the text classification, the GD metric is not reliable.
The reason is rooted in the architecture difference: the models are based on the pretrained LMs of more parameters than the conventional CNNs thus have more complex decision boundary and can be sensitive to the non-training samples.
We examined the evaded cases and found the cases of GD values higher than threshold are those with a small intersection ratio 0.2 and with more finetuning epochs.
Nevertheless, the OD metric guarantees the correct judgment.

\begin{tcolorbox}[notitle, boxrule=0.5pt,left=0.05cm, right=0.05cm, top=0.05cm, bottom=0.05cm]
\textbf{Takeaway 2:}
For the Attack~\ding{183}, OD and GD are robust against various training settings adopted by the adversary, and SD and DD can correctly quantify the regeneration ratio in the adversary's dataset.
\end{tcolorbox}

\begin{figure}[t]
    \centering
    \includegraphics[width=0.9\linewidth]{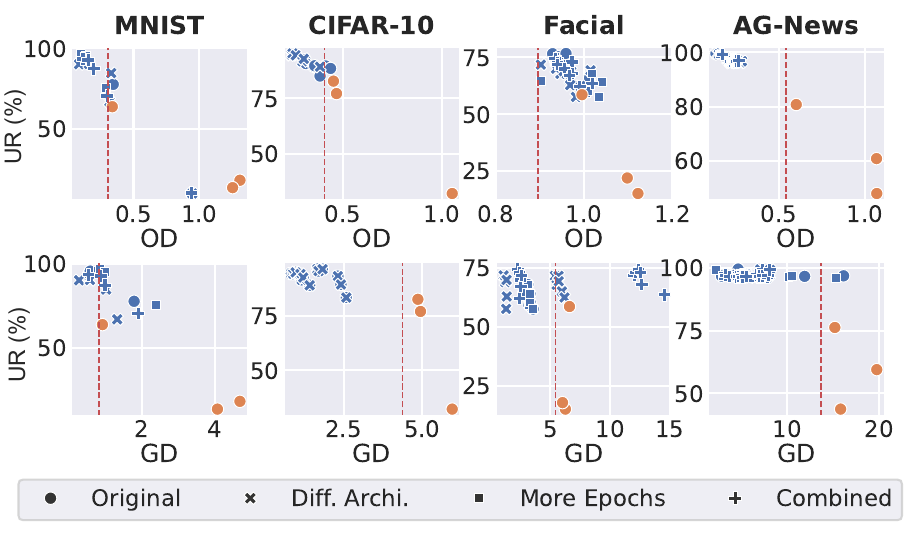}
    \caption{The trade-off between the UR and the model metrics (OD and GD) under different attack severity for Attack~\ding{184} and adversary model training settings. The blue and orange points are positive and negative datasets.}
    \label{fig:a3_rosbustness}
\end{figure}

\subsection{Defending Distribution-level Attack}
\label{subsec:exp_defend_attack3}
In this section, we investigate the effectiveness and robustness against the distribution-level dataset regeneration attack.
For vision tasks, we train generative models (see \Cref{tab:dataset}) and select checkpoints to obtain generated content of varying quality.
For the default generator, we train the EDM model on CIFAR-10 (size $32\times 32$) to an FID score of 2.341 and train StyleGAN3 on FairFace (size $128\times 128$) to an FID score of 12.151, which is comparable to the reported performance~\cite{DBLP:conf/nips/KarrasAAL22,DBLP:conf/nips/KarrasALHHLA21}.

\iheading{Limitation of OD and GD.} We provide the evaluation results of model-level metrics in \Cref{fig:a3_rosbustness} which are measurable when the defender has black-box or white-box model access and no access to partial suspect dataset.
As the adversary dataset is re-sampled from the same distribution as the source dataset, our model metrics can fail under different training settings.
We observe that multiple positive derivatives (blue points) can have both high UR and higher values of OD or GD than the negative ones at the same time.
Take the facial dataset as an example, a remarkable proportion of positive models achieves $\sim$75\% UR (the same level as Attack~\ding{182} in \Cref{fig:a1_rosbustness} and Attack~\ding{183} in \Cref{fig:a2_trade_off_odgd}) and has both model metrics higher than the threshold.
This is expected, because there are no exact overlapping samples and the training-data distribution shift between the positive models and the source model is minimized by the generative model.
Therefore, it is necessary for the dataset metric to identify the distribution-level regenerated dataset.

\begin{table}[t]
\caption{Robustness against Attack~\ding{184} with different generator types and generation sizes of the advanced generator.}
\label{tab:a3_robustness}
\resizebox{\linewidth}{!}{
\begin{tabular}{ccccccccc}
\toprule
\multirow{2}{*}{\textbf{\begin{tabular}[c]{@{}c@{}}Adversary\\ Setting\end{tabular}}} & \multicolumn{2}{c}{\textbf{MNIST}} & \multicolumn{2}{c}{\textbf{CIFAR-10}} & \multicolumn{2}{c}{\textbf{FairFace}} & \multicolumn{2}{c}{\textbf{AG-News}} \\ \cline{2-9} 
 & \multicolumn{1}{c}{UR (\%)} & DD & \multicolumn{1}{c}{UR (\%)} & DD & \multicolumn{1}{c}{UR (\%)} & DD & \multicolumn{1}{c}{UR (\%)} & DD \\ \midrule
\begin{tabular}[c]{@{}c@{}}Default\\ Generator\end{tabular} & \multicolumn{1}{c}{95.59} & \greenbox{1.43} & \multicolumn{1}{c}{92.87} & \greenbox{1.36} & \multicolumn{1}{c}{66.09} & \greenbox{2.52} & \multicolumn{1}{c}{96.75} & \greenbox{8.23} \\ \hline
\begin{tabular}[c]{@{}c@{}}Advanced\\ Generator\end{tabular} & \multicolumn{1}{c}{97.44} & \greenbox{1.18} & \multicolumn{1}{c}{95.63} & \greenbox{1.45} & \multicolumn{1}{c}{93.50} & \greenbox{2.53} & \multicolumn{1}{c}{97.40} & \greenbox{8.22} \\ \hline
Double size & \multicolumn{1}{c}{97.69} & \greenbox{0.94} & \multicolumn{1}{c}{97.29} & \greenbox{1.13} & \multicolumn{1}{c}{91.93} & \greenbox{2.06} & \multicolumn{1}{c}{97.63} & \greenbox{8.00} \\ \hline
Half size & \multicolumn{1}{c}{96.89} & \greenbox{1.18} & \multicolumn{1}{c}{93.39} & \greenbox{1.47} & \multicolumn{1}{c}{81.36} & \greenbox{2.68} & \multicolumn{1}{c}{96.34} & \greenbox{8.21} \\ \midrule
Threshold & \multicolumn{1}{c}{/} & 24.03 & \multicolumn{1}{c}{/} & 1.49 & \multicolumn{1}{c}{/} & 8.37 & \multicolumn{1}{c}{/} & 10.17 \\ \bottomrule
\end{tabular}
}

\end{table}

\iheading{Evaluation of SD and DD.}
Now we investigate whether the SD and DD, measurable under defender's access to partial suspect dataset, can track the distribution-level regeneration. 
As the dataset is regenerated, there are very few samples that resemble the original data samples, and thus the SD values are close to 0 for all the four benchmarks we selected.
As for DD, we evaluate the robustness by varying the attack settings and summarize the results in \Cref{tab:a3_robustness}.
In particular, we evaluate two different generative models for each benchmark.
For the vision task, we use GANs and diffusion models of appropriate scales: we use the StyleGAN2-Ada~\cite{Karras2020ada} for CIFAR-10 and the Denoising Diffusion Probabilistic Models (DDPM)~\cite{DBLP:conf/nips/HoJA20} for MNIST and FairFace. 
For the text, we use both back-translation and Pegasus-based paraphraser~\cite{DBLP:conf/icml/ZhangZSL20,pegasus_paraphraser} as advanced generator.
We also test if the adversary generates dataset of double or half size of the source dataset.

In \Cref{tab:a3_robustness}, we denote the generators that produce training data enabling model to achieve higher UR by ``Advanced'' and the other one as the ``Default'', and we report the generator checkpoint that achieves the highest UR.
Regardless of the generator type or the size of generated dataset, our DD can effectively and \textit{robustly} evaluate the closeness between the distributions of the regenerated dataset and the source dataset.
Next, we perform an in-depth analysis of the factors involved in distribution-level regeneration.

\iheading{Generator Capacity.}
The adversary may adopt a more advanced generator to produce more in-distribution samples.
The dataset produced by the advanced generator has higher quality but not necessarily lower DD value, because our DD assesses the closeness between the distribution of two datasets, instead of the dataset utility closeness.
We manually inspect the difference between generators and found that the better quality comes from finer details (\eg, clearer facial contours) in the generation, which can ameliorate the model convergence to the target distribution.
Nonetheless, the DD remains robust in identifying the regenerated dataset.

\iheading{Amount of Generated Samples.}
We note that having fewer generated samples for training is a limiting factor in developing a highly accurate suspect model.
Consequently, the Attack~\ding{184} may be preferred by the adversary because the adversary can further improve the model performance by increasing the training data size.
However, this comes at the cost of additional computation for generator training, finetuning, and especially inference.
For example, if the adversary's dataset size doubles, the training time also roughly doubles, not counting the cost of dataset generation.
Therefore, we consider doubled data size in this paper and leave larger sizes for future work.
Among the tested cases (\ie, double-size), our dataset metrics achieve robust tracking.

\begin{tcolorbox}[notitle, boxrule=0.5pt,left=0.05cm, right=0.05cm, top=0.05cm, bottom=0.05cm]
\textbf{Takeaway 3:}
\boxname can perform accurate dataset tracking with DD under stronger distribution-level attacks.
\end{tcolorbox}

\subsection{Comparison with Baselines}\label{subsec:exp_compare}
\begin{table}[t]
\centering
\caption{Comparison with existing methods. Each entry is TP/Total for the corresponding attack, and ``Overall TPR'' is computed over all attacks.}
\label{tab:comparison}
\resizebox{\linewidth}{!}{
\begin{tabular}{ccccccccc}
\toprule
\multirow{2}{*}{\textbf{Method}} & \multicolumn{4}{c}{\textbf{CIFAR-10}} & \multicolumn{4}{c}{\textbf{FairFace}} \\ \cline{2-9} 
 & A~\ding{182} & A~\ding{183} & A~\ding{184} & Overall TPR (\%) & A~\ding{182} & A~\ding{183} & A~\ding{184} & Overall TPR (\%) \\ \midrule
 Radioactive~\cite{DBLP:conf/icml/SablayrollesDSJ20} & 1 / 20 & 1 / 9 & 0 / 13 & 4.76 & 0 / 20 & 1 / 8 & 0 / 7 & 2.86 \\ \hline
 DI~\cite{maini2021dataset} & 1 / 20 & 2 / 9 & 0 / 13 & 7.14 & 0 / 20 & 1 / 8 & 0 / 7 & 2.86\\ \hline
 EMA~\cite{huang2022dataset} & 3 / 20 & 2 / 9 & 0 / 13 & 11.90 & 0 / 20 & 1 / 8 & 0 / 7 & 2.86 \\ \hline
UBW~\cite{li2022untargeted} & 4 / 20 & 3 / 9 & 0 / 13 & 16.67 & 0 / 20 & 1 / 8 & 0 / 7 & 2.86 \\ \hline
RAI$^2$~\cite{DBLP:conf/ndss/DongLCXZ023} & 6 / 20 & 5 / 9 & 0 / 13 & 26.19 & 18 / 20 & 7 / 8 & 0 / 7 & 71.43 \\ \hline
Data-use~\cite{datause_ccs2024}  & 15 / 20 & 7 / 9 & 0 / 13 & 52.38 & 18 / 20 & 8 / 8 & 0 / 7 & 74.29 \\ \hline
\textbf{Ours} & \textbf{20 / 20} & \textbf{9 / 9} & \textbf{13 / 13} & \textbf{100.0} & \textbf{20 / 20} & \textbf{8 / 8} & \textbf{7 / 7} & \textbf{100.0} \\ \bottomrule
\end{tabular}
}

\end{table}

\bheading{Setup.}
We compare with state-of-the-art methods including dataset watermarking UBW~\cite{li2022untargeted} and data-use audit~\cite{datause_ccs2024}, inference-based auditing EMA~\cite{huang2022dataset} and RAI$^2$~\cite{DBLP:conf/ndss/DongLCXZ023}.
We use CIFAR-10 and FairFace as they are two representative benchmarks used in prior work.
Since detection tasks typically prioritize a high TPR to minimize missed detections, we compare TPR, defined as $TPR=TP/(TP+FN)$, which also aligns with previous work~\cite{maini2021dataset,huang2022dataset,li2022untargeted,DBLP:conf/ndss/DongLCXZ023,datause_ccs2024} that mainly reports the detection rate.
For fair comparison, we need to unify the output format.
Note that UBW and EMA output binary auditing results that align with \boxname's output format, while RAI$^2$ estimates the proportion of copied samples in the form of a dataset similarity score ranging from 0 to 1.
We map the output to 1 if the dataset similarity is higher than 0.5, and 0 otherwise.
We use 0.5 as a midpoint for binarization; other thresholds trade off misses and false alarms and would require task-specific threshold selection.
In terms of testing samples, we add up to 20 sample-level modifications for Attack~\ding{182} to expand attack settings, select all possible overlap settings for Attack~\ding{183}, and consider datasets generated from various intermediate checkpoints of the generator for Attack~\ding{184}.
In total, we obtain 42 and 35 tested regenerated datasets for CIFAR-10 and FairFace, respectively.

\bheading{Comparison Results}.
In \Cref{tab:comparison}, we observe that \boxname systematically outperforms prior work across different threats and benchmarks.
The improvement comes from more robust detection against Attack~\ding{182} and Attack~\ding{183} and the coverage of Attack~\ding{184}.
For instance, \boxname can robustly detect sample-level attacks while previous watermarking (\eg, \cite{DBLP:conf/icml/SablayrollesDSJ20}) can be bypassed by certain sample-level perturbations (\eg, Glass blur) because of watermark loss caused by compression.
Notably, on FairFace, RAI$^2$~\cite{DBLP:conf/ndss/DongLCXZ023} and Data-use~\cite{datause_ccs2024} have comparable performance with our \boxname in detecting the sample-level and set-level attacks but fail to detect the distribution-level attack.
This signifies the necessity of a holistic approach to track dataset regeneration.

\begin{tcolorbox}[notitle, boxrule=0.5pt,left=0.05cm, right=0.05cm, top=0.05cm, bottom=0.05cm]
\textbf{Takeaway 4:}
\boxname outperforms previous approaches in terms of the detection rate through its multi-scale testing.
\end{tcolorbox}

\subsection{Robustness to Adaptive Attacks}\label{subsec:exp_robust_potential_attack}
\begin{table}[t]
\caption{Evaluation of DP-based adaptive attacks on MNIST. \redbox{Red}/\greenbox{green} indicates values above/below the threshold, respectively.}
\label{tab:adaptive}
\resizebox{\linewidth}{!}{
\begin{tabular}{ccccccccc}
\toprule
\multirow{2}{*}{\textbf{Attacks}} & \multicolumn{4}{c}{$\epsilon =1$} & \multicolumn{4}{c}{$\epsilon = 10$} \\ \cline{2-9} 
 & UR (\%) & OD & GD & DD & UR (\%) & OD & GD & DD \\ \midrule
DP-SGD~\cite{DBLP:conf/ccs/AbadiCGMMT016} & 76.15 & \greenbox{0.248} & \redbox{4.725} & / & 91.14 & \greenbox{0.109} & \redbox{1.637} & / \\ \hline
DP-MERF~\cite{DBLP:conf/aistats/HarderAP21} & 65.76 & \redbox{1.070} & \redbox{4.864} & \greenbox{2.577} & 71.78 & \redbox{1.213} & \redbox{5.224} & \greenbox{2.582} \\ \hline
PrivSet~\cite{DBLP:conf/nips/ChenKF22} & 73.36 & \redbox{1.323} & \redbox{9.482} & \greenbox{12.67} & 92.33 & \redbox{1.362} & \redbox{8.402} & \greenbox{11.73} \\ \bottomrule
\end{tabular}
}

\end{table}

The adversary can attack our framework from the model level and the dataset level with knowledge of \boxname's procedure, including the audit set selection and the metric design.
The main goal is to weaken the signals captured by \boxname.
To achieve this goal, we consider three types of adaptive attacks: 1) a universal adaptive attack based on differential privacy (DP) that reduces the influence of the original training data on the suspect model or dataset, 2) a metric-specific adaptive attack that optimizes against a given metric to raise the metric value above the threshold known by the adversary, and 3) an audit set-targeted adaptive attack where the adversary knows the audit set design and adjusts the suspect dataset before training or preview release.
We do not additionally combine these adaptive attacks with every base attack from Attack~\ding{182} to Attack~\ding{184} because such combinations introduce additional utility loss.
We mainly use MNIST and FairFace as they represent datasets of simple and complex features.

\bheading{Evaluation of DP-based Adaptive Attack.}
For the DP-based adaptive attack, the adversary can exploit DP-SGD~\cite{DBLP:conf/ccs/AbadiCGMMT016} for model-level metrics OD and GD.
For our SD and DD metrics, the adversary generates data with DP for the distribution-level attack in order to increase the difference between regenerated and protected datasets.
We consider MNIST here because existing DP generation techniques mainly focus on simple benchmarks such as MNIST.
For DP-generated data, we use DP-MERF~\cite{DBLP:conf/aistats/HarderAP21} to train a generator that can produce data of better quality (\ie, enabling training models of higher test accuracy) than conventional DP GANs.
In addition, the adversary can directly distill the target dataset into a much smaller dataset using dataset distillation techniques~\cite{DBLP:conf/wacv/ZhaoB23}.
We adopt PrivSet~\cite{DBLP:conf/nips/ChenKF22} as the most recent dataset condensation solution with DP constraints.
\Cref{tab:adaptive} shows the results of DP-based attacks against our framework \boxname.

\iheading{Model-level DP-based Attack.}
As expected, DP applied during training on MNIST has a negative impact on utility (\ie, test accuracy).
The small budget $\epsilon=1$ can reduce the UR to $\sim70\%$, while the large budget $\epsilon=10$ only results in a slight UR drop; it still achieves UR over 90\%, compared to UR$\approx98\%$ under Attack~\ding{182}.
According to \Cref{alg:judge}, our framework \boxname still flags these cases because at least one decisive model metric (OD or GD) remains below the determined threshold.
This suggests that the DP-based attack against OD/GD does not enable evasion without incurring a noticeable utility loss in our setting.

\iheading{Dataset-level DP-based Attack.}
We evaluate the DP-based attack against our remaining dataset metrics.
As the bottom two rows of \Cref{tab:adaptive} show, the regenerated datasets are not identified as the sample-level attack or the set-level attack according to our framework \boxname because the OD and GD values are higher than the threshold.
In addition, similar to the evaluation in \Cref{subsec:exp_defend_attack3}, the SD values are almost 0 because there are few approximate sample pairs between the original and regenerated datasets.
Hence, we rely on DD to distinguish whether the regenerated dataset has a nearly identical distribution.
As \Cref{tab:adaptive} shows, the DD values are below the threshold, demonstrating that the dataset-level attack does not evade \boxname.
Compared to the MNIST results of \Cref{tab:a3_robustness}, the DD values of dataset-distillation-generated datasets (\ie, PrivSet) are higher than those previously evaluated with na\"ive generators, further confirming that a smaller regeneration dataset size can increase DD at the cost of degraded UR.

\bheading{Metric-specific Adaptive Attacks.}
With knowledge of the metric design and the target threshold (\ie, estimated from the design of \boxname), the adversary can add a metric-maximization term weighted by $w$ to the optimization objective to increase the metric value for their trained model or generated dataset.
When the defender can only compute OD and GD metrics, the adversary minimizes the model-training loss minus $w\cdot d(\vicf, \advf, \auditset^{\cA})$, where $d\in\{OD, GD\}$ and $\auditset^{(M, \cA)}$ is an adversary-generated audit set following the same procedure as $\auditset^M$.
As the adversary does not know the actual $\auditset^M$ due to random sampling, the $\auditset^{(M, \cA)}$ is sampled in each step.
The adversary can adopt better hyperparameters than the default setting, so we applied a grid search on optimizers, training epochs from 10 to 200, learning rates $\{0.001,0.01, 0.1\}$, and finally selected training of 200 epochs with learning rate 0.1 using Adam.
For the DD-specific attack (\ie, improving Attack~\ding{184}), we mainly focus on MNIST because of lower training cost.
We use conditional GAN on MNIST for experiments, and add DD as an optimization term with different penalty weights $w$.
We apply similar hyperparameter optimization on the generator training.

\begin{figure*}[t]
    \centering
    \includegraphics[width=0.95\linewidth]{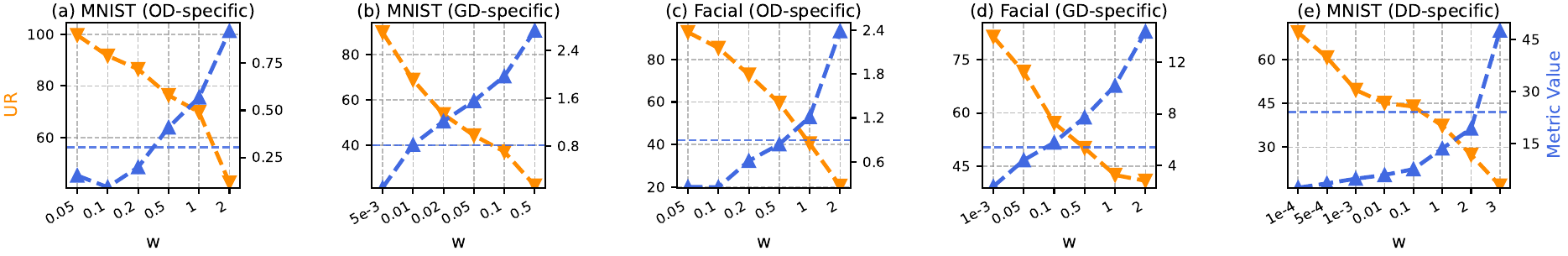}
    \caption{Evaluation of metric-specific adaptive attacks.}
    \label{fig:metric-specific}
\end{figure*}

\Cref{fig:metric-specific} (a) and \Cref{fig:metric-specific} (c) show the OD-specific adaptive attacks on MNIST and FairFace, respectively.
We can see that small $w$ is already effective to increase the auditing metrics and large penalty weight $w$ can significantly degrade the model utility.
The sweet spot between the utility and the auditing metric value differs between MNIST and FairFace.
For instance, on MNIST the adversary begins to bypass OD-only auditing around $w=0.5$, but the UR has already dropped below 80\% and decreases further as $w$ grows.
Meanwhile, on FairFace, $w=0.5$ retains about 60\% UR but remains near the OD threshold, while larger weights cross the threshold at substantially lower UR because the facial data are of larger scale and make the model more sensitive to the metric penalty.

For the GD-specific adaptive attacks on MNIST and FairFace in \Cref{fig:metric-specific} (b) and \Cref{fig:metric-specific} (d), we observe a similar trend but the model's performance is more sensitive to the GD penalty, as UR drops faster than under the OD penalty.
This is expected because the GD penalty can have a larger impact on the model weights during training through gradients of larger norms.
To make GD higher than our predefined threshold, the UR is around 70\% and less than 60\% for MNIST and FairFace respectively, suggesting that GD is harder to attack adaptively.
Meanwhile, it increases the adversary's computational cost because GD depends on per-example parameter gradients, and optimizing it builds a second-order computation graph which slows the training and occupies more memory.

For the DD-specific adaptive attack shown in \Cref{fig:metric-specific} (e), we found that the UR values are generally lower than those in previous OD-specific and GD-specific adaptive attacks on MNIST.
Specifically, in this case, the generative model is more sensitive to the penalty weight and often fails to converge stably to a low loss.
Therefore, the final generator produces low-quality images, which makes the data unlearnable for the subsequent classification models.
The results demonstrate a trade-off between the adaptive attack strength (\ie, $w$) and the final model or dataset utility: increasing the metric-specific penalty enough to bypass detection comes with utility loss.

\begin{table}[t]
\caption{Evaluation of audit set-targeted adaptive attacks on FairFace.}
\label{tab:adaptive-audit-set}
\resizebox{0.8\linewidth}{!}{
\begin{tabular}{@{}cccc@{}}
\toprule
\textbf{Target}        & \textbf{Strategy}             & \textbf{UR} & \textbf{Judgment}            \\ \midrule
\multirow{3}{*}{OD/GD} & 10\% high-loss sample removal & 86.86\%     & Attack~\ding{182}/\ding{183} \\
                       & 20\% high-loss sample removal & 85.71\%     & Attack~\ding{182}/\ding{183} \\
                       & 30\% high-loss sample removal & 84.24\%     & Independent                  \\
                       \midrule
\multirow{3}{*}{DD} & 10\% Low-loss sampled preview       & -           & Attack~\ding{184}            \\
                       & 10\% High-loss sampled preview      & -           & Independent                  \\
                       & Single-class preview          & -           & Independent                  \\ \bottomrule
\end{tabular}
}
\end{table}

\bheading{Adaptive Attacks for Audit Set.}
We consider two adaptive strategies with knowledge of the design of $\auditset^M$ and $\auditset^D$, targeting the model-level metrics OD/GD and dataset-level metrics SD/DD, respectively.
The adversary cannot control $\auditset^M$ but can remove potential audit-set samples by filtering out samples with loss higher than $\tau_{audit}$.
Since $\tau_{audit}$ is a hyperparameter set by the defender and should be adjusted based on the adversary's own model, one direct method for the adversary is to remove high-loss samples.
Then, the adversary directly trains $\advf$ on the remaining data.
The top half of \Cref{tab:adaptive-audit-set} shows that removing 30\% of the samples can lead to an independent judgment by \Cref{alg:judge} when the audit is limited to model-level metrics (OD/GD).
Further exploration on the removed data shows that our previous $\tau_{audit}$ corresponds to the top 24\%, and removing less than 20\% high-loss samples is equivalent to Attack~\ding{183}, which is detectable as shown in \Cref{fig:a2_trade_off_odgd}.
Although removing 30\% can evade model-level tracking with OD/GD, without knowledge of $\tau_{audit}$, the adversary cannot reliably guess the correct removal ratio to win the game in \Cref{subsec:formulation} and still risks utility loss or being tracked.

In addition, we consider an adversary that regenerates with Attack~\ding{184} and controls $\auditset^D$ (\eg, an adversary-controlled preview subset on the platform), which is sampled from regenerator-produced data of 10\% low-loss or high-loss samples assessed by a surrogate model directly trained on the obtained $\advdst$, or from a random single class.
In the bottom half of \Cref{tab:adaptive-audit-set}, we find that DD can still track the low-loss case but fails for single-class and high-loss cases because of significant distribution shift.
Although the adversarial selection of $\auditset^D$ does not degrade $\advdst$, the high-loss samples can contain low-quality generated data, and a single class does not fully represent $\advdst$ well enough to attract users.

\begin{tcolorbox}[notitle, boxrule=0.5pt,left=0.05cm, right=0.05cm, top=0.05cm, bottom=0.05cm]
\textbf{Takeaway 5:}
The evaluated adaptive attacks against \boxname can evade detection by increasing metric values or manipulating the audit set, but they incur dataset utility degradation, poorer preview representativeness, or memory and computational cost.
\end{tcolorbox}

\section{Related Work}
\label{sec:related}
In this section, we review existing copyright protections and auditing strategies for models and datasets.

\bheading{Models.}
The owner can watermark the proprietary model before release and testify the watermark to claim ownership, but this type of watermarking can be vulnerable to various adaptive attacks~\cite{DBLP:conf/sp/LukasJLK22}.
Similarly, white-box watermarking, which manipulates the weights with watermarking order in a functionality-preserving way, is also shown unreliable~\cite{DBLP:conf/uss/YanP0023} under adversarial settings.
Another direction is leveraging the inherent model patterns for ownership verification.
For example, PublicCheck~\cite{wang2022publiccheck} realizes public verification of a deployed model at runtime through fingerprint samples.
RAI$^2$~\cite{DBLP:conf/ndss/DongLCXZ023} verifies model identity using random projection without the need for fingerprints or watermarks.
DeepJudge~\cite{DBLP:conf/sp/ChenWPS0JM0S22} implements a testing framework of six metrics for model similarity testing.
Our work also informs model ownership disputes by auditing dataset usage, which directly affects model performance.

\bheading{Datasets.}
There is also a line of work attempting to track dataset usage~\cite{du2024sok}.
Dataset watermarking~\cite{DBLP:conf/icml/SablayrollesDSJ20,yimingli_tifs23} refers to marking datasets with backdoor-alike samples.
Yet, the robustness is questionable under white-box access and low-sample settings~\cite{DBLP:journals/corr/abs-2202-12506}.
On the other hand, non-invasive auditing can better preserve the dataset utility.
Dataset Inference~\cite{maini2021dataset} tests whether the model is trained on the source dataset and is recently extended to the LLM datasets~\cite{maini2024llm}.
EMA~\cite{huang2022dataset} infers the dataset identity through enhanced membership inference to the suspect model.
RAI$^2$~\cite{DBLP:conf/ndss/DongLCXZ023} uses a look-up table to estimate dataset similarity via black-box querying of the suspect model.
ORL-AUDITOR~\cite{dataset_auditing_rl} audits trajectory usage to protect the copyright of offline reinforcement learning datasets. 
Nevertheless, prior methods largely provide binary evidence of dataset use from shallow signals (\eg, memorization) and typically do not aim to attribute the underlying regeneration mechanism, especially for distribution-level polishing.
Our framework covers more advanced regeneration attacks and provides holistic auditing with coarse-grained attribution between shallow-feature regeneration (Attack~\ding{182}/\ding{183}) and distribution-level polishing (Attack~\ding{184}).
In addition to the sample-, user-, and dataset-level categories surveyed in prior work~\cite{du2024sok}, we further study distribution-level dataset tracking.

\section{Discussion \& Conclusion}
\label{sec:discussion}
In this work, we propose \boxname, a multi-scale testing framework to track regenerated datasets.
We validate the effectiveness and robustness of our framework on vision and text classification tasks.

\bheading{Limitation \& Future Work.}
\boxname currently assumes white-box model access and partial dataset access to enable full-scale testing with all metrics. This assumption is realistic for platform-based audits, but it limits the applicability of \boxname in offline transactions where the auditor cannot inspect the suspect model or dataset.
In the black-box regime, tracking regenerated datasets is particularly challenging under distribution-scale perturbations, because there may be no near-duplicate samples whose memorization can be reliably detected from model outputs alone. As future work, we plan to explore deep-learning interpretability techniques to extract and compare learned representations from black-box models, which may provide stronger signals for tracking regenerated datasets.

\bheading{Scalability.}
Although \boxname is currently designed for domain-specific datasets that are typically small, scaling to larger datasets used to train classification or generative models is also feasible.
For large classification datasets, we can adopt a divide-and-conquer strategy.
Specifically, we can measure local similarity between the suspect and source datasets and aggregate the results into a global similarity assessment.
The downstream judgment process should also be adapted accordingly.
For large datasets used to train generative models, we can, for example, adapt OD and GD to generative models by replacing the output with latent representations because the model's internal representations can also memorize training data~\cite{ganleak}.

\bibliographystyle{ACM-Reference-Format}
\bibliography{ref}

\appendix

\section{Ethical Considerations}

\bheading{Context and Goal.}
This work studies dataset-regeneration (piracy) threats and proposes \boxname, a multi-scale testing framework that helps an auditor assess whether a suspect dataset is regenerated from a protected victim dataset under sample-, set-, and distribution-level regeneration primitives.
Our experiments are conducted offline on datasets and models obtained from open-source platforms (\eg, Hugging Face) or prior literature, under the permissions and licenses specified by those sources.
We obtained IRB approval from the authors' institutions and did not recruit or interact with human participants.

\bheading{Stakeholders.}
We consider the following stakeholders potentially impacted by our research:
\begin{enumerate}
    \item \emph{Dataset creators or owners} (\ie, victim in the paper) whose proprietary or scarce datasets may be pirated;
    \item \emph{Model developers or data users} who may be accused of piracy (rightly or wrongly);
    \item \emph{Platforms and auditors} (\eg, model or dataset hosting services) who may deploy auditing pipelines;
    \item \emph{Data subjects} represented in datasets, especially sensitive domains (\eg, facial or biometric);
    \item \emph{End users and society} who benefit from accountable AI supply chains and reduced illicit reuse;
    \item \emph{Adversaries} who may attempt to evade auditing or misuse released artifacts.
\end{enumerate}

\bheading{Potential Benefits.}
\boxname can improve transparency and accountability of dataset usage by enabling non-intrusive detection signals beyond exact-sample matching, including distribution-level evidence, which can support dispute resolution and deter dataset piracy. It may also incentivize better governance for scarce and high-value datasets.

\bheading{Potential Harms and Rights Considerations.}
Any tracking result may be misinterpreted as definitive proof of wrongdoing, potentially harming legitimate users through reputational or legal consequences.
Adversaries could use the insights to tune regeneration procedures or evade detection, potentially increasing piracy success at high cost.
Although we do not collect new personal data, some evaluated benchmarks may involve human-related content (\eg, FairFace dataset). Mishandling or over-sharing derived artifacts could amplify privacy risks.
Dataset licenses and terms of use vary; irresponsible redistribution of datasets or derivative artifacts could violate legal constraints.

\bheading{Mitigations.}
We position \boxname as an \emph{auditing support tool} rather than a standalone legal verdict.
Therefore, conclusions should be corroborated with additional evidence and human review.
We limit the release of materials that would directly enable dataset piracy at scale (see details in the Open Science section), while still providing sufficient artifacts to evaluate our defensive methodology.
We do not probe or test live systems or services without consent, and we do not perform vulnerability exploitation on third-party infrastructure.
We rely on publicly available datasets and models and follow their usage constraints.
We do not attempt to re-identify individuals, and we avoid publishing any new sensitive annotations or linking information beyond what is already present in the original sources.

\bheading{Decision Rationale.}
Balancing benefits (\ie, deterrence and accountability for dataset piracy, and improved auditing capabilities) against
harms (\ie, dual-use and misinterpretation risks), our work is ethically justified \emph{with mitigations}: careful release choices, and avoidance of new human-subject data collection.
We will revisit the release of additional artifacts only if and when doing so does not materially increase misuse risk, and if appropriate safeguards are available.

\section{Open Science}

We release the implementation of \boxname~\footnote{\url{https://zenodo.org/records/20746861}}, including the end-to-end auditing pipeline and a user-friendly interface, together with environment specifications and installation instructions to facilitate reproducibility.
To reduce misuse risk, we do not release turnkey implementations of the three dataset-regeneration attacks studied in this paper.
Nevertheless, we provide sufficient methodological details to re-implement these attacks.
Moreover, all evaluations are conducted on publicly available datasets; thus, interested users can readily re-create the experimental setup and test the described attack primitives on open datasets without requiring any proprietary dataset.

\section{Generative AI Usage}

We used GPT-5.2 and GPT-5.4 as a writing assistant to improve clarity and presentation (\eg, polishing wording, reorganizing paragraphs, and improving readability) during manuscript revision.
The tool was \emph{not} used to generate experimental data, run experiments, produce figures, or create new technical claims.

\section{Additional Experimental Details}\label{appendix:additional_detail}

\subsection{Measurement Details}
\label{appendix:measurementdetails}
The verification process was conducted by three authors. Initially, two authors compared suspected duplicate samples filtered by $L_2$ distance to confirm whether they were identical. The criterion for judging two samples as exact duplicates is that they represent the same content and details. We use an $L_2$ filter with threshold 0.1 as a preliminary step; zero $L_2$ distance indicates exact duplicates. The manual comparison aimed to exclude cases of small non-zero $L_2$ distance where discrepancies arose due to noise, image post-processing (\eg, overexposure), or incomplete images (\eg, nearly blank images).

For the inter-annotator agreement protocol, if the two evaluators reached conflicting conclusions, a third author with a background in AI and medicine resolved the disagreement by determining whether the tumor images were identical. This process ensured the reliability of the manual verification and the reproducibility of the duplication findings.

\subsection{Details and Additional Evaluation}
\label{appendix:details_and_results}

\begin{figure}[t]
    \centering
    \includegraphics[width=0.95\linewidth]{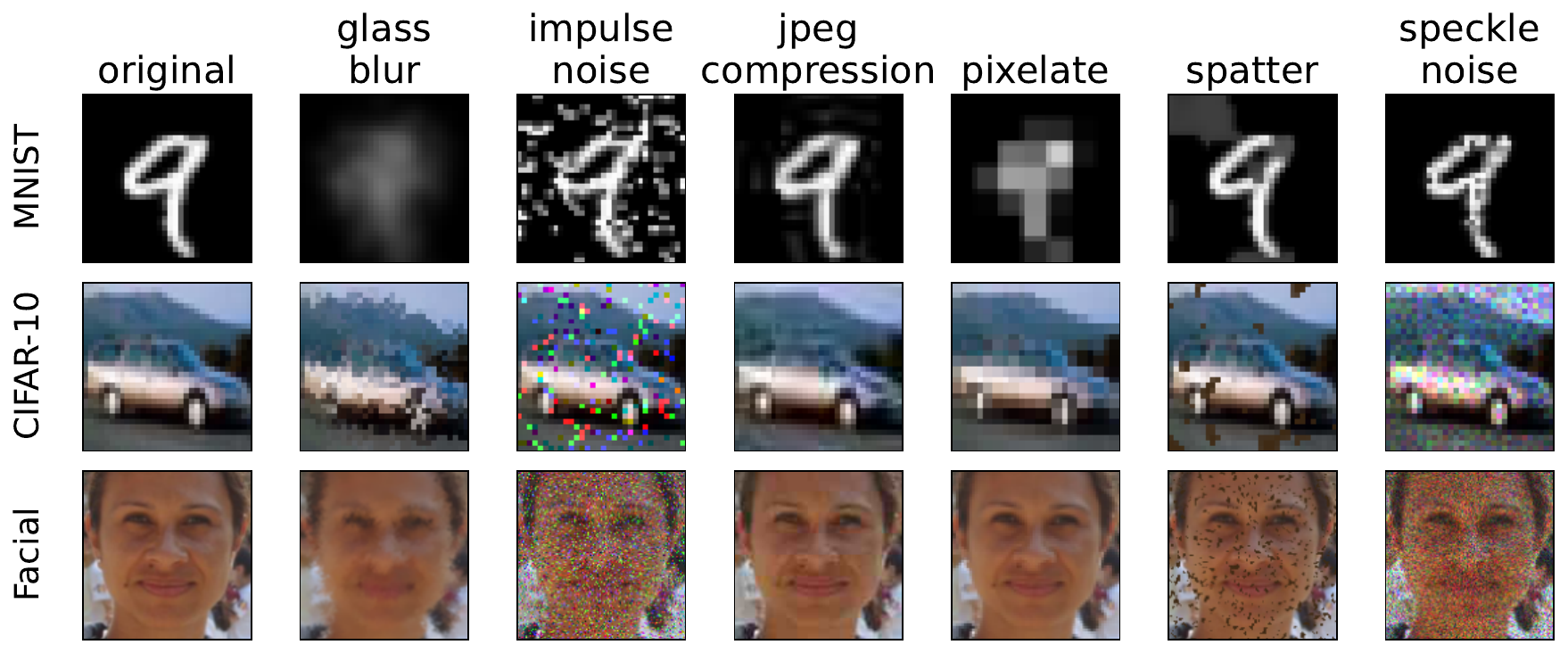}
    \caption{Examples of sample-level regeneration attack.}
    \label{fig:a1_visualize}
\end{figure}

\begin{figure}[t]
    \centering
    \includegraphics[width=0.95\linewidth]{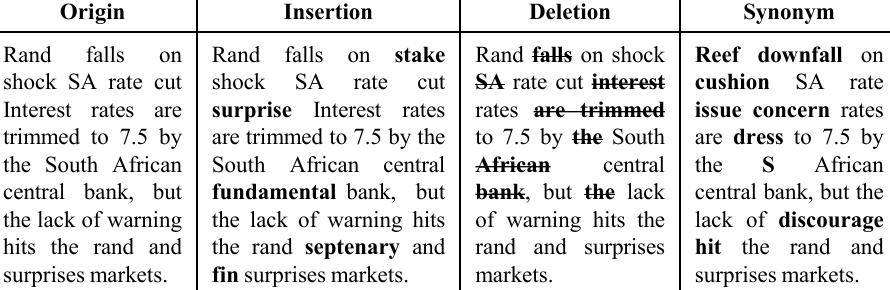}
    \caption{Text examples of sample-level regeneration attack.
    The modified words are highlighted.
    }
    \label{fig:example_agnews}
\end{figure}

\bheading{Evaluation on tumor datasets.}
We evaluated \boxname on the tumor datasets measured in the paper and on additional open datasets collected after our initial measurement study. In total, we gathered 13 datasets.
In our initial analysis, \boxname reports signals consistent with potential post-processing (Attack \ding{182}) and exact duplication (Attack \ding{183}) among 10 datasets, and no such signals in the remaining 3 datasets. This aligns with our measurement results.

\bheading{Default Training.}
In the default setting, we train models for 200 epochs with learning rate 0.1 and with SGD to ensure convergence.
For the text, the default finetuning setting for the pretrained smaller BERT models is 3 epochs with learning rate 5e-4.

\bheading{Metric Setting.}
We set $L_2$ norm to measure the data distance in the metrics.
As for the metric SD, we set the closed ball's radius $r=16/255=0.0625$ for image data and $2.77$ as the $L_2$ norm of a full-$0.1$ embedding.
For DD, we instantiate $\psi_\theta$ as a two-layer ReLU MLP with hidden dimension 128 and output dimension 128, with all parameters (weights and biases) sampled i.i.d.\ from $\mathcal{N}(0,1)$ and kept fixed (not trained).
For vision tasks, the MLP input is the raw pixel vector (flattened image); for text tasks, the MLP input is the BERT \texttt{[CLS]} embedding as described in \Cref{subsec:expsetup}.
Following the selection process of prior work~\cite{DBLP:conf/sp/ChenWPS0JM0S22}, we use the default threshold ratio $\alpha_\lambda=0.9$ for $\lambda\in\{OD,GD,DD\}$.

\begin{figure}[t]
    \centering
    \includegraphics[width=0.8\linewidth]{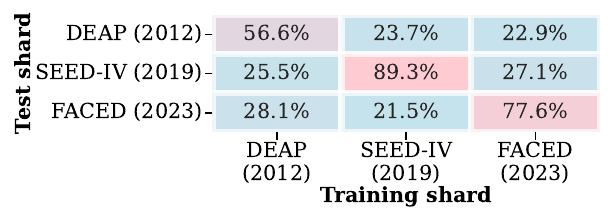}
    \caption{An example experiment of transfer accuracy between different EEG feature datasets, each of which is split into training and test shards. For each training shard, we train a model and test across three test shards. Low transfer accuracy indicates the importance of in-distribution dataset for specific tasks.}
    \label{fig:transfer-eeg}
\end{figure}

\bheading{Experiments on EEG Datasets.}
We justify the dataset importance for domain-specific task through an example of emotion classification via Electroencephalography (EEG) data in \Cref{fig:transfer-eeg}.
We surveyed and chose three classic and latest representative (licensed) EEG \textit{feature} benchmarks: DEAP~\cite{koelstra2011deap}, SEED-IV~\cite{zheng2018emotionmeter} and FACED~\cite{chen2023large}.
Although we manually unify labels, the transfer accuracy scores are low, possibly because of unbalanced labeling and different collections, indicating that, given test data, training accurate models requires its in-distribution data, which may not be accessible to external developers and thus worth protection.
In \Cref{fig:transfer-eeg}, we test transferred accuracy of emotion classification using ResNet-18, where we observe the consistent degradation of data distribution discrepancy, indicating the uniqueness of dataset for building exclusive AI products.

\bheading{Combined Attacks.}
Without loss of generality, we consider the combination of shallow and deep feature attacks: the adversary applies Attack~\ding{182} or \ding{183} to the dataset generated by the ``Advanced generator'' trained on MNIST (\Cref{tab:a3_robustness}).
We apply high severity of post-processing in \Cref{tab:a1mod_all_distance} for Attack~\ding{182}.
Among the tested post-processing variants, Speckle noise yields the highest UR (93.57\%) and Glass blur yields the lowest UR (69.81\%), with DD values 3.14 and 11.12, respectively.
Both DD values are under the threshold of MNIST, indicating the Attack~\ding{182} is detectable.
However, in this case our judgment cannot reliably distinguish the two shallow-feature attack types (Attack~\ding{182} vs.\ \ding{183}).
For Attack~\ding{183}, we tested ratio $\{0.6,0.8\}$ which leads to UR values 95.44\% and 96.63\% with DD values 13.79 and 6.38, respectively. 
Therefore, the simple combination cannot increase the measured distances without degrading data utility.

\section{Theoretical Analysis}
\label{appendix:theory}
We provide detailed justifications for OD, GD and DD.

\subsection{Justification of OD}
\label{app:od}

We formalize why OD is informative under a standard stability route. We emphasize that stability is a sufficient (not necessary) condition for generalization; we adopt it to obtain a transparent bound.
We first introduce the uniform stability~\cite{bousquet2002stability}.
\begin{definition}[Uniform stability]
A learning algorithm $L$ has uniform stability $\beta>0$ w.r.t.\ loss $\ell$ if for any two
datasets $S,S'$ of size $n$ that differ in at most one example,
\begin{equation}
\sup_{z}\, \big|\ell(L(S),z)-\ell(L(S'),z)\big|\ \le\ \beta.
\end{equation}
\end{definition}

A classical result links uniform stability to expected generalization gap for bounded losses.
We use it as motivation; our OD bound below relies on prediction stability.

\begin{lemma}[Expected generalization gap under stability]
\label{lem:gen-gap}
Assume $\ell\in[0,1]$ and $L$ has uniform stability $\beta$. Then
\begin{equation}
\big|\mathbb{E}_{S}\big[ R(L(S)) - R_S(L(S))\big]\big|\ \le\ \beta,
\end{equation}
where $R(\cdot)$ is population risk and $R_S(\cdot)$ is empirical risk on $S$.
\end{lemma}

\bheading{Prediction stability.}
Uniform stability is defined on loss and does not automatically upper-bound output differences
$\|f_S(x)-f_{S'}(x)\|$ without additional assumptions.
We therefore introduce a prediction-stability
condition.
This condition is standard in stability-based analyses for regularized learning algorithms~\cite{hardt2016train} (\eg, DNN training), and we treat it as an explicit modeling assumption for our setting.

\begin{definition}[Prediction stability]
\label{def:pred-stab}
Let $f_S=L(S)$ denote the learned predictor. We say $L$ is \emph{prediction-stable} with parameter
$\beta_f$ if for any adjacent datasets $S,S'$ (differing in one example),
\begin{equation}
\sup_{x}\, \|f_S(x)-f_{S'}(x)\|\ \le\ \beta_f.
\end{equation}
Intuitively, $\beta_f$ upper-bounds the worst-case \emph{output perturbation} (measured in the same norm as OD) induced by replacing a single training example; it depends on the learning algorithm, regularization, optimization dynamics, and early stopping.
\end{definition}

\iheading{OD bound for $k$-sample regeneration.}
Let $\vicdst$ be the victim dataset and $\advdst$ be a regenerated dataset that differs by at most $k$ samples.
Construct a sequence $S_0,\dots,S_k$ such that $S_0=\vicdst$, $S_k=\advdst$, and $S_{i+1}$ differs from $S_i$ in one example.
Let $f_i=L(S_i)$.

\begin{proposition}[OD scales at most linearly with $k$ under prediction stability]
\label{prop:od-linear}
Assume $L$ is prediction-stable with parameter $\beta_f$ (\Cref{def:pred-stab}). Then for any audit set
$X_J$,
\begin{equation}
\frac{1}{|X_J|}\sum_{x\in X_J} \|f_\cV(x)-f_\cA(x)\|
\ \le\ k\,\beta_f,
\end{equation}
where $f_\cV=L(\vicdst)$ and $f_\cA=L(\advdst)$.
Consequently, $\mathrm{OD}(f_\cV,f_\cA,X_J)=O(k)$.
\end{proposition}

\begin{proof}
By triangle inequality,
\begin{equation}
\|f_\cV(x)-f_\cA(x)\|=\|f_0(x)-f_k(x)\|\le \sum_{i=0}^{k-1}\|f_i(x)-f_{i+1}(x)\|.
\end{equation}
Applying \Cref{def:pred-stab} to each adjacent pair gives $\|f_i(x)-f_{i+1}(x)\|\le \beta_f$ for all $x$,
hence $\|f_\cV(x)-f_\cA(x)\|\le k\beta_f$. Averaging over $X_J$ yields the claim.
\end{proof}

Proposition~\ref{prop:od-linear} shows that OD is upper-bounded by a term linear in $k$, the number of training points that differ between $\vicdst$ and $\advdst$.
For Attack~\ding{183}, if the adversary preserves an overlap ratio $p$ with the victim set, then $k \approx (1-p)\lvert \vicdst\rvert$, which explains the near-linear OD trend observed in \Cref{fig:a2_rosbustness}.
This also clarifies why achieving simultaneously small OD and high utility requires keeping the effective perturbation size $k$ limited, motivating our complementary GD/DD checks for stronger adversaries.

\subsection{Justification of GD}
\label{app:gd}

We connect GD to influence-function analysis and clarify the assumptions needed for a rigorous interpretation.
For twice-differentiable empirical risk minimization with optimum $\hat\theta$, the (upweighting) influence of a
training point $z_{\text{train}}$ on the loss at a test point $z_{\text{test}}$ admits the classical approximation:
\begin{equation}
I_{\mathrm{up}}(z_{\text{train}},z_{\text{test}})
\ \approx\
\nabla_\theta \ell(z_{\text{test}},\hat\theta)^\top
H_{\hat\theta}^{-1}
\nabla_\theta \ell(z_{\text{train}},\hat\theta),
\end{equation}
where $H_{\hat\theta}$ is the Hessian of the empirical objective at $\hat\theta$.

\iheading{Self-influence and a curvature-controlled upper bound.}
The self-influence of a point $z$ is often summarized by the quadratic form
\begin{equation}
I_{\mathrm{self}}(z)\ :=\ \nabla_\theta \ell(z,\hat\theta)^\top H_{\hat\theta}^{-1}\nabla_\theta \ell(z,\hat\theta).
\end{equation}

The $LNG(f,x)$ metric, and by extension GD, is a computationally efficient, Hessian-free surrogate
for a sample's \emph{self-influence}.
Under the classical influence-function approximation, the (self-)influence
of a sample $z=(x,y)$ around the optimum $\hat{\theta}$ takes the quadratic form
$I_{\mathrm{self}}(z)=g(z)^\top H_{\hat{\theta}}^{-1} g(z)$, where $g(z)=\nabla_{\theta}\ell(z;\hat{\theta})$ and
$H_{\hat{\theta}}$ is the Hessian of the empirical objective.
Since $LNG(f,x)$ measures the magnitude of per-sample parameter gradients across layers (with scale normalization),
it serves as a Hessian-free proxy for $g(z)$.

\begin{proposition}[Gradient norm controls self-influence under strong convexity / local curvature]
\label{prop:gd-influence}
Assume $H_{\hat\theta}\succeq \mu I$ for some $\mu>0$. Then for any point $z$,
\begin{equation}
I_{\mathrm{self}}(z)\ \le\ \frac{1}{\mu}\,\big\|\nabla_\theta \ell(z,\hat\theta)\big\|_2^2.
\end{equation}
\end{proposition}

\begin{proof}
Since $H_{\hat\theta}\succeq \mu I$, we have $H_{\hat\theta}^{-1}\preceq \frac{1}{\mu}I$.
Therefore $g^\top H_{\hat\theta}^{-1} g \le \frac{1}{\mu} g^\top g$ with $g=\nabla_\theta \ell(z,\hat\theta)$, because $\lambda_{\min}(H_{\hat{\theta}})\ge \mu>0$ implies $\lambda_{\max}(H_{\hat{\theta}}^{-1})\le 1/\mu$.
\end{proof}

A sample $(x,y)$ that was part of the training set has, by definition, contributed to the learned parameters $\hat{\theta}$. 
However, for modern over-parameterized and non-convex models, it is generally \emph{not} guaranteed that every training sample satisfies $\|\nabla_\theta \ell(f_{\hat{\theta}}(x),y)\|\approx 0$ at convergence. 
What we can rely on is a weaker and standard stationarity implication: at a (local) stationary point of ERM, the \emph{aggregate} gradient over the training set vanishes (or is small), and well-fitted samples typically exhibit smaller per-example gradients in practice.

Therefore, GD should be interpreted as a \emph{diagnostic proxy}: a low GD value on an audit point $x$ indicates that $\advf$ treats $x$ as a low-gradient point under its learned representation, which is \emph{consistent with} (but does not mathematically guarantee) $x$ being ``seen-like'' data for $\advf$.
This is sufficient for our detection goal, where GD is used as one component in a multi-scale evidence set rather than as a stand-alone proof of membership.

\subsection{Justification of DD}
\label{app:dd}
In our implementation, DD is computed on randomized representations $\psi_\theta(\mathbf{x})$.
Specifically, $\psi_\theta$ is a two-layer ReLU MLP with hidden dimension 128 and output dimension 128, whose parameters $\theta$ are sampled i.i.d.\ from $\mathcal{N}(0,1)$ and kept fixed; each draw of $\theta$ yields one random neural feature map. For vision tasks, the input is the flattened raw pixel vector; for text tasks, the input is the BERT \texttt{[CLS]} embedding (Appendix~\ref{appendix:additional_detail}).

DD (Eq.~\eqref{eq:dd}) is a randomized \emph{mean-embedding discrepancy}: it compares the empirical feature means of $\psi_\theta(\cdot)$ under two datasets and averages over random draws of $\theta$.
This is motivated by the Maximum Mean Discrepancy (MMD), which is defined as a distance between distributions via differences of mean embeddings in an RKHS.
Specifically, given two probability distributions $P$ and $Q$, and a Reproducing Kernel Hilbert Space (RKHS) $\mathcal{H}$ with kernel $k(\cdot, \cdot)$, the squared MMD is the distance between the mean embeddings of the distributions in $\mathcal{H}$:
\begin{equation}
MMD^2(P, Q) = \|\mu_P - \mu_Q\|^2_{\mathcal{H}},
\end{equation}
where $\mu_P = \expectedvalue_{x \sim P}[k(x, \cdot)]$ is the mean embedding of $P$. This can be expanded as:
\begin{equation}
    MMD^2(P, Q) = \expectedvalue_{x,x' \sim P}[k(x,x')] + \expectedvalue_{y,y' \sim Q}[k(y,y')] - 2\expectedvalue_{x \sim P, y \sim Q}[k(x,y)]
\end{equation}

When $\psi_\theta$ is chosen as a random-feature map for a characteristic kernel (\eg, random Fourier features for Gaussian RBF), DD recovers a standard random-feature approximation of MMD and inherits its identifiability properties in the limit.
Our instantiation uses random neural feature maps as a practical surrogate while preserving the key structural principle: matching distributions implies matching feature expectations.
The theoretical motivation comes from the following theorem about MMD.

\begin{theorem}[Identifiability of MMD under characteristic kernels~\cite{DBLP:journals/jmlr/SriperumbudurGFSL10,DBLP:journals/jmlr/GrettonBRSS12}]
\label{thm:mmd_characteristic}
Let $k$ be a characteristic kernel on $\mathcal{X}$ with RKHS $\mathcal{H}_k$.
Then the kernel mean embedding $\mu_P := \mathbb{E}_{x\sim P}[k(x,\cdot)]\in\mathcal{H}_k$ is injective, and
\begin{equation}
\mathrm{MMD}_k(P,Q) \;=\; \|\mu_P-\mu_Q\|_{\mathcal{H}_k} \;=\; 0
\quad\Longleftrightarrow\quad
P=Q .
\end{equation}
\end{theorem}

\begin{proof}[Proof sketch]
By definition, $\mathrm{MMD}_k(P,Q)=\|\mu_P-\mu_Q\|_{\mathcal{H}_k}$.
Thus $\mathrm{MMD}_k(P,Q)=0$ iff $\mu_P=\mu_Q$.
For characteristic kernels, the mean embedding map $P\mapsto \mu_P$ is injective, i.e.,
$\mu_P=\mu_Q \Rightarrow P=Q$.
Therefore $\mathrm{MMD}_k(P,Q)=0$ iff $P=Q$.
\end{proof}

\iheading{Implication for distribution-level regeneration.}
A successful ``polishing'' attack aims to make the adversary distribution $\cP_\cA$ match the victim distribution $\cP_\cV$,
i.e., $\cP_\cA \approx \cP_\cV$.
In particular, if $\cP_\cA=\cP_\cV$, then for any fixed $\theta$ we have $\mathbb{E}_{x\sim \cP_\cA}[\psi_\theta(x)]=\mathbb{E}_{x\sim \cP_\cV}[\psi_\theta(x)]$, so the population version of DD is $0$.
Therefore, the more the adversary succeeds in matching $\cP_\cV$, the smaller DD tends to be, increasing the likelihood of being flagged by \boxname.

\subsection{Derivation of the Detection-Utility Trade-off Bound}
\label{app:bound}
Here, we provide the formal derivation for \Cref{thm:tradeoff}. We model the adversary's problem using the standard framework of domain adaptation.
\begin{itemize}
    \item \textbf{Source Domain ($S$):} The adversary's distribution $\mathcal{P}_{\mathcal{A}}$.
\item \textbf{Target Domain ($T$):} The victim's distribution $\mathcal{P}_{\mathcal{V}}$.
\end{itemize}
The adversary trains on the source domain, minimizing $R_{S}(h)=\expectedvalue_{(x,y)\sim \mathcal{P}_{\mathcal{A}}}[\ell(h(x),y)]$,
while utility is evaluated on the target domain $R_{T}(h)=\expectedvalue_{(x,y)\sim \mathcal{P}_{\mathcal{V}}}[\ell(h(x),y)]$.

The key to bounding the relationship between these two risks is the $\mathcal{H}\Delta\mathcal{H}$-divergence, which measures the dissimilarity between two distributions $P$ and $Q$ relative to a hypothesis class $\mathcal{H}$.
\begin{definition}
The $\mathcal{H}\Delta\mathcal{H}$-divergence between two distributions $P$ and $Q$ is:
$$
d_{\mathcal{H}\Delta\mathcal{H}}(P, Q) = 2 \sup_{h, h' \in \mathcal{H}} |\Pr_{x \sim P}[h(x) \ne h'(x)] - \Pr_{x \sim Q}[h(x) \ne h'(x)]|
$$
This quantity measures the maximum disagreement between the two domains over the regions where any two hypotheses in $\mathcal{H}$ disagree.
\end{definition}

\bheading{Risk Definition.}
For clarity, we state the standard binary-classification form with the $0$--$1$ loss.
For a hypothesis $h\in\mathcal H$, define the source/target risks as
\begin{equation}
R_S(h) \triangleq \Pr_{(x,y)\sim \cP_\cA}\left[h(x)\neq y\right],\quad
R_T(h) \triangleq \Pr_{(x,y)\sim \cP_\cV}\left[h(x)\neq y\right].
\end{equation}
Since under the $0$--$1$ loss and for any hypothesis $h\in\mathcal{H}$~\cite{ben2010theory}:
\begin{equation}
R_T(h) \le R_S(h) + \tfrac{1}{2}\, d_{\mathcal{H}\Delta\mathcal{H}}(\cP_\cA, \cP_\cV) + \lambda,
\end{equation}
where $\lambda = \min_{h^* \in \mathcal{H}} (R_S(h^*) + R_T(h^*))$ is the \textit{ideal joint error}, representing the minimum achievable error on both domains by an optimal hypothesis $h^*$.
Substituting our domain definitions, we get \Cref{thm:tradeoff}:
\begin{equation}
R_{\mathcal{P}_{\mathcal{V}}}(h) \le R_{\mathcal{P}_{\mathcal{A}}}(h) + \tfrac{1}{2}\, d_{\mathcal{H}\Delta\mathcal{H}}(\mathcal{P}_{\mathcal{A}}, \mathcal{P}_{\mathcal{V}}) + \lambda
\end{equation}
\textbf{Implication:} This bound formalizes a utility--divergence trade-off:
\begin{itemize}
    \item $R_{\mathcal{P}_{\mathcal{V}}}(h)$ is the adversary's utility loss (they want this to be low).
    \item $R_{\mathcal{P}_{\mathcal{A}}}(h)$ is the risk on their own data (they can minimize this via training).
    \item $d_{\mathcal{H}\Delta\mathcal{H}}(\mathcal{P}_{\mathcal{A}}, \mathcal{P}_{\mathcal{V}})$ captures a hypothesis-disagreement notion of domain shift between the adversary and victim distributions.
    \item $\lambda$ is the ideal joint error, determined by how well the best hypothesis in $\mathcal{H}$ can perform simultaneously on both domains.
\end{itemize}
The inequality shows that as cross-domain discrepancy increases,
the upper bound on the victim-domain risk becomes looser, making high-utility transfer harder to guarantee.
Therefore, to maintain high utility on the victim-side distribution, the adversary generally needs to keep the cross-domain discrepancy controlled, which constrains how far the regenerated distribution can drift to evade distribution-level detection in our evaluated settings.

\subsection{False Positive and Power Analysis}
\label{app:fp_power}

We now analyze the final judgment as a statistical test.
For each metric $\lambda\in\{OD,GD,DD\}$, let $S_\lambda$ be the measured value.
The null hypothesis $H_0$ is that the suspect dataset is independently constructed and the alternative hypothesis $H_1$ is that the suspect dataset is regenerated from $\vicdst$.
Since smaller distances indicate higher similarity to $\vicdst$, the rejection region is $S_\lambda\leq\tau_\lambda$.

\bheading{False Positive \& Threshold.}
For a fixed task and a single metric, the false-positive rate is $\Pr[S_\lambda\leq\tau_\lambda\mid H_0]$.
The threshold $\tau_\lambda=\alpha_\lambda\cdot\min_i s^i_{neg}$ defined in \Cref{subsec:judgement} uses the closest available negative score for that task.
If the $m$ negatives used to set the threshold and a future independent dataset for the same task are drawn in the same way from the same population, the future negative is equally likely to take any rank among the $m+1$ scores and therefore falls below all $m$ negatives with probability $1/(m+1)$.
When $\alpha_\lambda<1$, the threshold is smaller than the closest negative score, making the test more conservative under the same assumptions.
This bound is limited when only a few negatives are available and does not apply when those negatives do not represent future honest datasets, such as common-source datasets close to $\vicdst$.
Algorithm~\ref{alg:judge} combines multiple metrics. Therefore, the overall false positive probability can be larger than the false positive probability of a single metric. A direct union-bound upper bound gives $\Pr[\mathrm{FP}]\leq \Pr[S_{OD}\leq\tau_{OD}\mid H_0]+\Pr[S_{GD}\leq\tau_{GD}\mid H_0]+\Pr[S_{DD}\leq\tau_{DD}\mid H_0]$ over the decisive distance metrics available to the auditor. This bound explains why each metric threshold should be conservative. It also explains why the judgment should be treated as technical evidence rather than a standalone proof.

\begin{table}[t]
\centering
\caption{Leave-one-negative-out threshold sensitivity based on the negative datasets in \Cref{tab:a1mod_all_distance,tab:a1mod_all_distance_agnews}. Thresholds are set separately within each task, and the false positives are aggregated over four tasks.}
\label{tab:false_positive_checks}
\resizebox{0.6\linewidth}{!}{
\begin{tabular}{cccc}
\toprule
\textbf{Threshold ratio $\alpha_\lambda$} & \textbf{FP/Total} & \textbf{FPR (\%)} & \textbf{TNR (\%)} \\ \midrule
0.80 & 3/12 & 25.0 & 75.0 \\
0.90 & 5/12 & 41.7 & 58.3 \\
0.95 & 6/12 & 50.0 & 50.0 \\
1.00 & 8/12 & 66.7 & 33.3 \\ \bottomrule
\end{tabular}}
\end{table}

To empirically understand how $\alpha_\lambda$ affects false positives, in \Cref{tab:false_positive_checks}, we test threshold sensitivity by varying $\alpha_\lambda$ for the negative datasets in \Cref{tab:a1mod_all_distance,tab:a1mod_all_distance_agnews}.
For each task, we leave one negative out and use the other two negatives from the same task to set the thresholds.
In total, we have 12 cases aggregating three held-out choices for each of the four tasks.
At the default $\alpha_\lambda=0.9$, the 5 leave-one-out false positives come from the closest held-out negatives (MNIST Neg-3, CIFAR-10 Neg-2, FairFace Neg-1, and AG-News Neg-1/Neg-3).
Specifically, MNIST Neg-3 and AG-News Neg-3 are triggered by model-level metrics, while CIFAR-10 Neg-2, FairFace Neg-1, and AG-News Neg-1 are triggered by DD.
The results with $\alpha_\lambda\in\{0.80,0.90,0.95,1.00\}$ show that larger $\alpha_\lambda$ is more permissive and increases FPR while reducing TNR.
We also check the positive side on \Cref{tab:a3_robustness}: after setting each task's DD threshold from its three negatives, $\alpha_\lambda=0.80$ would miss three DD positives in CIFAR-10 (13/16 detected), while $\alpha_\lambda\geq0.90$ detects all 16 DD positives.
Together, we observe the threshold-ratio trade-off: a smaller $\alpha_\lambda$ can reduce false positives but miss positives, while reliability still depends on the number and representativeness of task-specific negatives.

\bheading{Power analysis.}
The power of a metric test is $\Pr[S_\lambda\leq\tau_\lambda\mid H_1]$.
Suppose metric $\lambda$ separates regenerated datasets from the negatives used to set the threshold.
Let $\mu_0$ and $\mu_1$ be the expected metric values under independent and regenerated datasets, with $\mu_1<\mu_0$. Under a normal or delta-method approximation, the power is approximately $\Phi((\tau_\lambda-\mu_1)/(\sigma_1/\sqrt{n_{\mathrm{eff}}}))$, where $\sigma_1^2$ is the variance under regenerated datasets and $n_{\mathrm{eff}}$ is the effective number of audited measurements. When $\mu_1<\tau_\lambda$, this approximation suggests that power increases with larger positive-negative separation, larger audit-set or random-feature sample size, and lower metric variance. This formula is used to explain the trend rather than to provide an exact guarantee.
For OD, this means that a larger gap between the OD values of regenerated and negative models, together with more audit samples, makes the OD test more likely to trigger on regenerated datasets.
For GD, the interpretation applies when regenerated and independent datasets produce separated gradient statistics, since GD is a diagnostic proxy rather than a per-sample guarantee.
For DD, this trend is consistent with its mean-embedding interpretation in \Cref{app:dd}: more audited samples or random feature draws reduce the variance of the empirical mean discrepancy when the population DD values are separated.
For SD, $|\auditset^D|\cdot SD$ behaves like a count of matched audit samples under a binomial approximation, so the probability of observing enough matches increases with the true overlap ratio and with $|\auditset^D|$. In \Cref{alg:judge}, SD mainly provides supporting evidence for sample-level and set-level regeneration after the model metrics trigger the judgment.
Overall, the power analysis suggests that increasing the audit-set size or the number of random feature draws reduces estimation noise and improves detection power, while close negatives or small audit sets can make regenerated datasets harder to separate from independent ones.

\end{document}